\documentclass[twocolumn]{IEEEtran}
\usepackage{tabularx}
\usepackage[T1]{fontenc}
\usepackage{algorithmic}
\usepackage{graphicx}
\usepackage{amsmath,amssymb,amsfonts}
\usepackage{array}
\usepackage{textcomp}
\usepackage[font=small]{caption}
\usepackage{paralist}
\usepackage{lettrine}
\usepackage{amsthm}
\usepackage{cite}
\usepackage{xcolor}
\usepackage{cancel}
\usepackage{float}
\usepackage{booktabs}
\usepackage{subfig}
\usepackage{stfloats}
\usepackage{url}
\usepackage{lipsum}
\usepackage{enumitem}
\usepackage{mathtools}
\usepackage{cuted}
\usepackage{makecell}
\usepackage{acronym}
\usepackage{multicol}
\usepackage{multirow}
\usepackage[ruled,vlined]{algorithm2e}
\usepackage{authblk}
\usepackage{fixltx2e}
\usepackage{tikz}
\usetikzlibrary{arrows}

\newcommand*\curvedrightleft{%
  \mathrel{%
    \begin{tikzpicture}[scale=0.5]%
      \draw[->] (0,0.05) to [out=30,in=150] (1,0.05);%
      \draw[->] (1,-0.05) to [out=-150,in=-30] (0,-0.05);%
    \end{tikzpicture}%
  }%
}

\newtheorem{theorem}{Theorem}

\newcommand{\pdv}[2]{\frac{\partial #2}{\partial #1}}

\newcommand{\RN}[1]{%
  \textup{\uppercase\expandafter{\romannumeral#1}}%
}

\def\BibTeX{{\rm B\kern-.05em{\sc i\kern-.025em b}\kern-.08em
    T\kern-.1667em\lower.7ex\hbox{E}\kern-.125emX}}

\begin{document}

\title{Wireless Communications with Space-Time Modulated Metasurfaces}

\author{Marouan~Mizmizi,~\IEEEmembership{Member,~IEEE,}  Dario~Tagliaferri,~\IEEEmembership{Member,~IEEE,}   Umberto~Spagnolini,~\IEEEmembership{Senior Member,~IEEE}
\thanks{M.\, Mizmizi, D.\ Tagliaferri, U.\ Spagnolini are with the  Department of Electronics, Information and Bioengineering (DEIB) of Politecnico di Milano, 20133 Milan, Italy  (e-mail: [marouan.mizmizi, dario.tagliaferri, umberto.spagnolini]@polimi.it }
\thanks{U.\ Spagnolini is Huawei Industry Chair.}
\thanks{This paper was supported by the European Union under the Italian National Recovery and Resilience Plan (NRRP) of NextGenerationEU, partnership on “Telecommunications of the Future” (PE00000001 - program “RESTART”).}}

\maketitle 

\begin{abstract}
Space-time modulated metasurfaces (STMMs) are a newly investigated technology for the next 6G generation wireless communication networks. An STMM augments the spatial phase function with a time-varying one across the meta-atoms, allowing for the conveyance of information that possibly modulates the impinging signal. Hence, STMM represents an evolution of reconfigurable intelligent surfaces (RIS), which only design the spatial phase pattern. 
STMMs convey signals without a relevant increase in the energy budget, which is convenient for applications where energy is a strong constraint.
This paper proposes a mathematical model for STMM-based wireless communication, that creates the basics for two potential STMM architectures. One has excellent design flexibility, whereas the other is more cost-effective. The model describes STMM's distinguishing features, such as space-time coupling, and their impact on system performance. The proposed STMM model addresses the design criteria of a full-duplex system architecture, in which the temporal signal originating at the STMM generates a modulation overlapped with the incident one.
The presented numerical results demonstrate the efficacy of the proposed model and its potential to revolutionize wireless communication.
\end{abstract}

\begin{IEEEkeywords}
Space-time modulated metasurfaces, space-time phase coupling, RIS, 6G
\end{IEEEkeywords}

\section{Introduction}

The development of future wireless communication networks is a collective effort that entails the integration of multiple technologies as well as addressing several challenges. The goal of 6G is to provide ultra-high data rates, low latency, and high reliability for a wide range of applications and services such as autonomous vehicles, eHealth, and the Internet of Things \cite{SaaBenChe:J20, mizmizi20216g}. To reach these objectives, 6G links must integrate advanced technologies with the use of high-frequency bands and energy-efficient solutions.

The emergence of electromagnetic (EM) metasurfaces, a.k.a. reconfigurable intelligent surfaces (RIS), has opened up new possibilities in the design of wireless propagation \cite{di2020smart}. A RIS is typically made of quasi-passive arrays of sub-wavelength-sized meta-atoms whose scattering properties are engineered to manipulate the impinging EM waves and control the reflection \cite{9999292, zhang2021performance}.
RIS underwent tremendous conceptual and, to a minor extent, experimental development through several studies. In particular, RISs have been proposed for overcoming blockage \cite{Zeng9539048}, with specific application to vehicular scenarios \cite{Mizmizi2022_conformal,Tagliaferri2022_conformal,Mizmizi2022_conformal_conf}, for coverage extension \cite{Kurt9359529} and, recently, as support for sensing~\cite{9732186}.

The RIS technology (and metamaterials in general) has experienced a surge in interest in the last decade, following the formulation of the generalized Snell's law in 2011 \cite{doi:10.1126/science.1210713}. It is based on the conservation of both the linear momentum and the energy of reflected and refracted waves \cite{shaltout2015time}. The linear momentum is related to spatial symmetry, i.e., an EM wave impinging from an arbitrary angle $\theta_i$ is reflected at the specular one $\theta_o = \theta_i$. Setting the spatial-phase gradient across the RIS implies breaking the spatial symmetry, thus allowing a reflection toward $\theta_o\neq \theta_i$. Energy conservation relates to the time-reversal symmetry or Lorentz reciprocity, i.e., if we can reflect an EM wave coming from a particular angle $\theta_i$ towards a specific angle $\theta_o$, the same EM wave coming back from $\theta_o$ will be reflected toward $\theta_i$. 

\begin{figure*}[t]
    \centering
    \subfloat[RIS]{\includegraphics[width=0.32\textwidth]{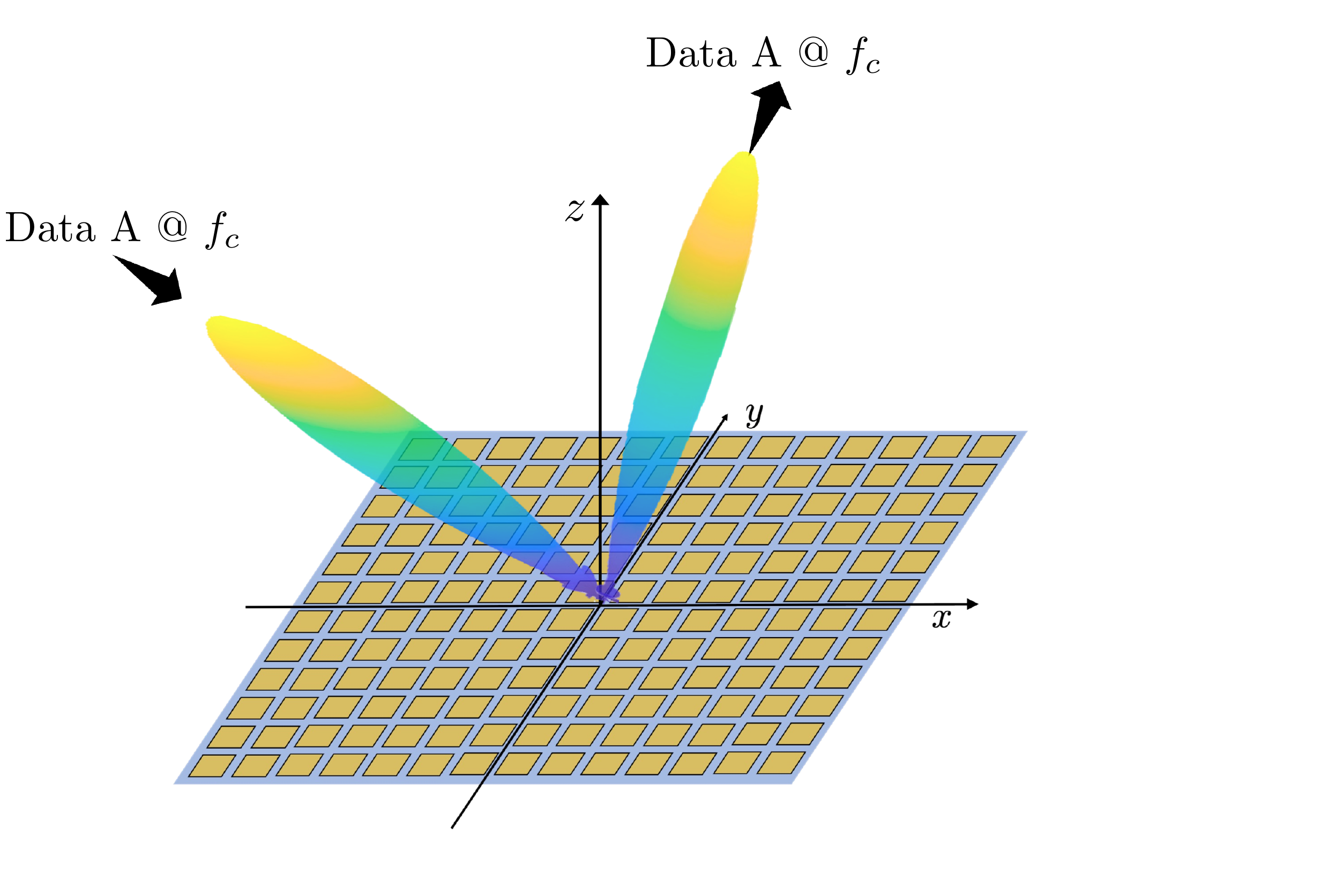}\label{sub_fig:Intro_RIS}}
    \subfloat[STCM]{\includegraphics[width=0.32\textwidth]{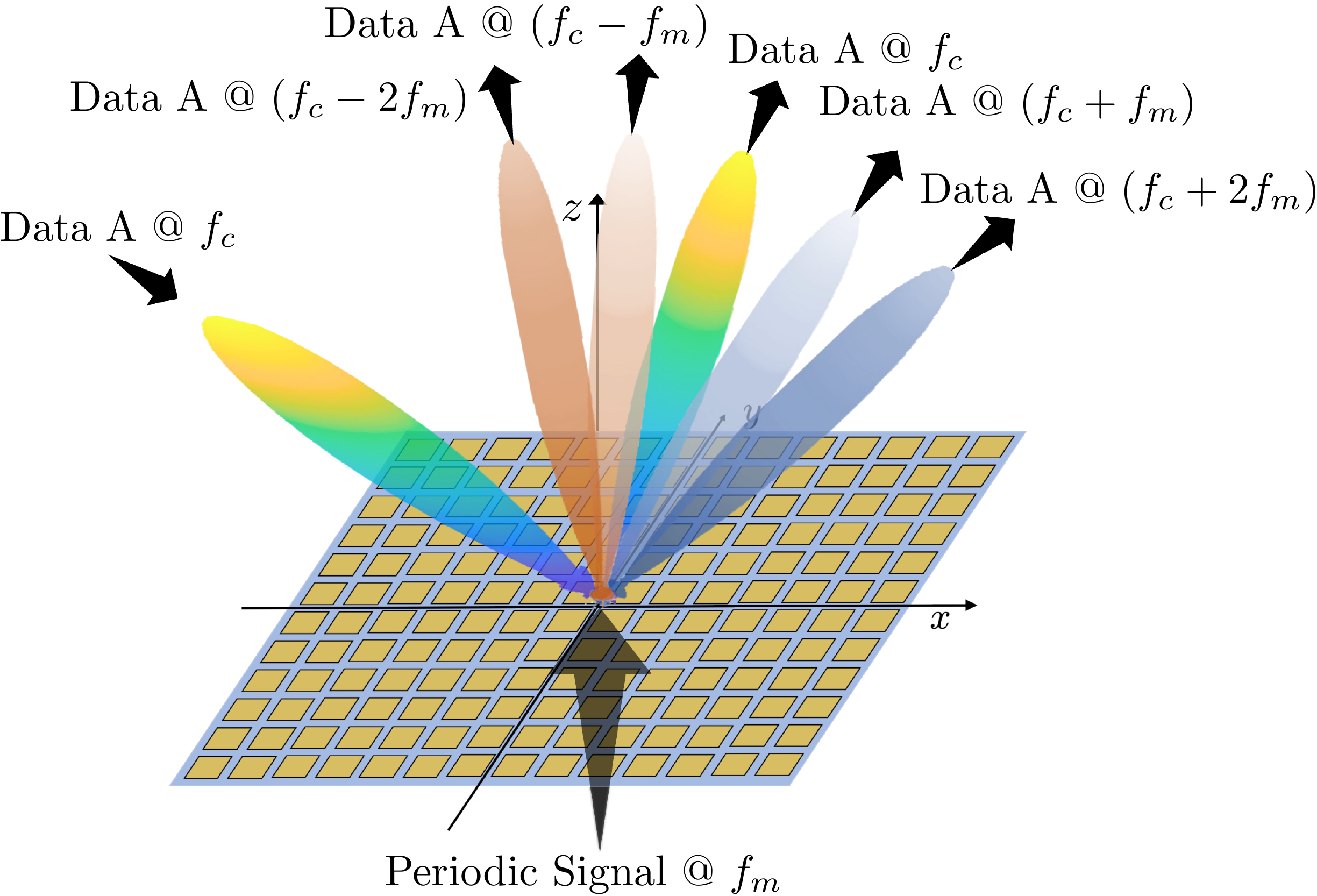}\label{sub_fig:Intro_STCM}}\hspace{0.01\textwidth}
    \subfloat[Proposed STMM]{\includegraphics[width=0.32\textwidth]{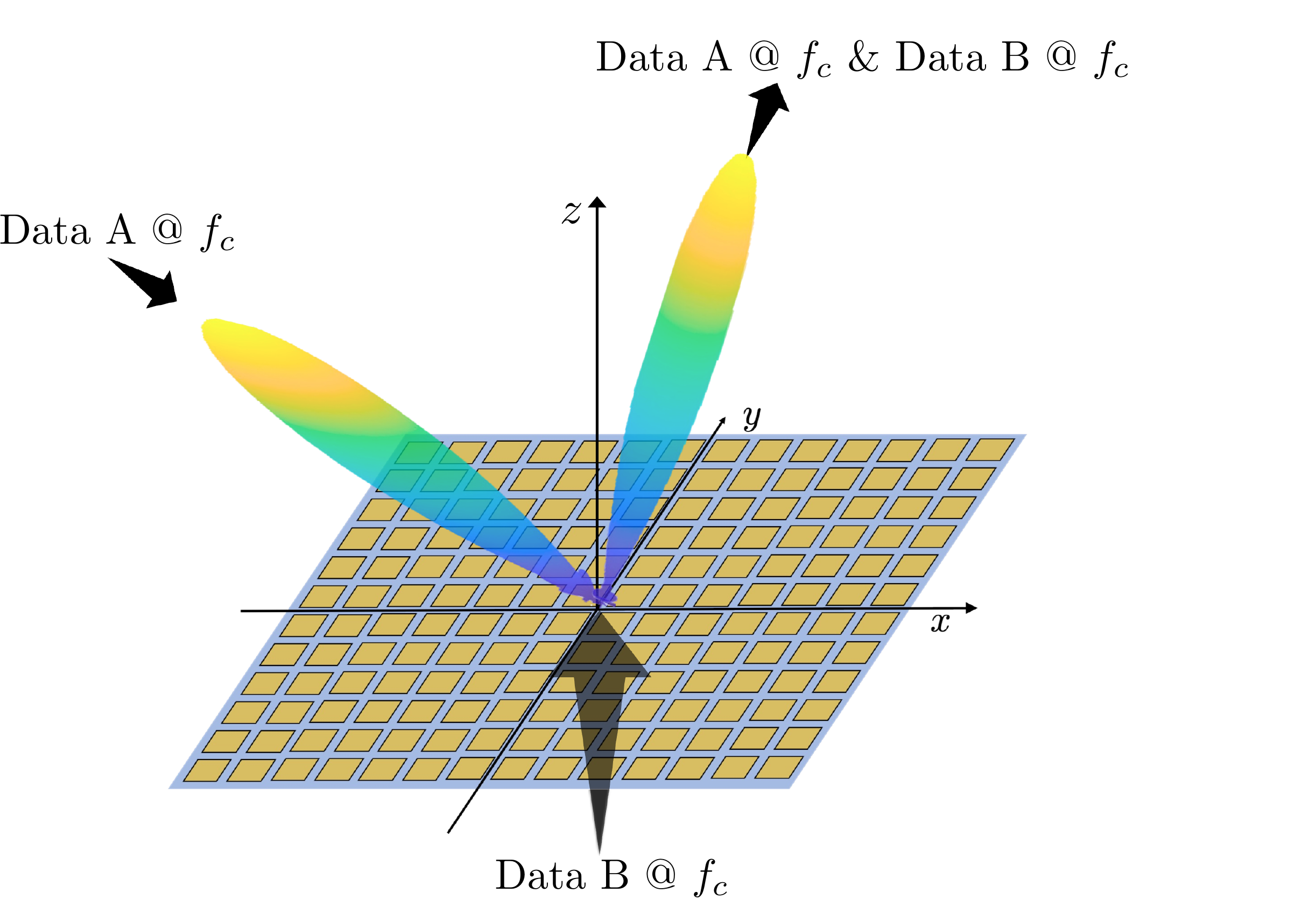}\label{sub_fig:Intro_STMM}}\hspace{0.01\textwidth}
    \caption{The application of metasurfaces has evolved from (\ref{sub_fig:Intro_RIS}) RIS, which allows control of the direction of reflection, to (\ref{sub_fig:Intro_STCM}) STCM, which allows more advanced reflection and control of the reflected frequency, to (\ref{sub_fig:Intro_STMM}) STMM, which can convey information superposed to the reflected signal.}
    \label{fig:IntroFigs}
\end{figure*}

In 2014, the authors in \cite{hadad2015space} proposed the \textit{universal Snell's law} of reflection realizing efficient solar panels. The key idea is to break the energy conservation law by adding a temporal phase gradient across the metasurface and then combining the spatial and temporal phase gradients in the metasurface to create anomalous nonreciprocal reflections. This principle is at the foundation of space-time modulated metasurfaces (STMMs) and generalizes RISs, opening up the possibility of conveying information in the reflected EM waves \cite{taravati2022microwave}. Fig. \ref{fig:IntroFigs} depicts the evolution of metasurfaces from RIS to STCM and STMM.
The authors in \cite{IJ133_NatComm_9_4334_2018, Zhang2020, Zhang2021, GengBo_STCM2022} used a spatial-temporal gradient to create space-time coded metasurfaces (STCMs), which can both reflect the signal in space and shift its frequency, generating harmonics.
Herein, STMMs differ from STCMs in that they are intended to convey information through reflection, whereas STCMs are primarily designed to synthesize time-varying reflection patterns.
In this latter direction, the work in \cite{9133266} proposes an STMM-based MIMO transmitter that operates on a sinusoidal feed and encodes data using the harmonics of a properly designed phase-discontinuous symbol. However, the interaction between spatial and temporal phases has never been investigated. Nevertheless, the theoretical findings of \cite{9133266} are validated by experimental results in a $2 \times 2$ MIMO setting, demonstrating a rate of 20 Mbps. The authors of \cite{8901437} discuss the realization of advanced modulation schemes by engineering the phase response of the single STMM meta-atom.
From a system-level perspective, a wireless communication system using a time-modulated metasurface to convey information can be conceptually considered a backscatter communication system, where a backscatter device (\textit{tag}) transmits information to an intended receiver by reflecting the signal impinging from an external transmitter \cite{Liang2022_RIS_backscatter_survey}. The authors of \cite{Alouini2023_LoRa} outline the potential of backscatter communications for low-power Internet of Things applications, providing analytical tools to evaluate the system performance for a generic backscattering tag (with possibly amplification capabilities). The generalized space shift keying for ambient backscatter communication systems is proposed in \cite{Tellambura2022_ambientbackscatter}, where the RIS switches the reflected signal between two receiving points, encoding the information in the direction of reflection. 
Among the relevant contributions using RISs, the authors of \cite{Park2020} propose to enhance ambient backscatter communications with a phase-modulated RIS employing a binary phase shift keying modulation. The work \cite{Jingtao2022_IRS_backscattering} proposes a RIS-based backscatter communication system to realize computational task offloading in energy-constrained networks, harvesting part of the energy of the impinging signal to enable the conveyance of information through reflection. 
The work in \cite{9516949}, instead, describes a RIS-based beamforming-plus-phase modulation system to serve a user with information augmentation, focusing on beamforming and detection performance.

All of the preceding works either discuss the physical realization or detail the specific applications of STCMs/STMMs (referred to as time-modulated RISs in backscatter communication literature), but never explicitly account for space-time phase coupling. By generalizing the RIS principle, stemming from the physical principle ruling the space-time modulation of metasurfaces, i.e., the universal Snell's law, this paper aims to demonstrate the potential and associated challenges of an STMM-based wireless communication system. Herein, we consider a master unit (MU), for example, a base station (BS), and a slave unit (SU), e.g., user equipment (UE), in a generic full-duplex system (MU $\curvedrightleft$ SU). The phase control of STMM enables the superposition of a MU $\leftarrow$ SU signal to a MU $\rightarrow$ SU one originating at the MU.
Full-duplex technologies and systems are of great interest because they will enable many 6G services, and the potential applications and challenges are discussed in \cite{6832464}. In the proposed study case, the link MU $\rightarrow$ SU uses a portion of the bandwidth $B_d$ in the spectrum. The SU uses the STMM to \textit{retro-reflect and modulate} the received MU $\rightarrow$ SU signal with a phase-only MU $\leftarrow$ SU signal with a bandwidth $B_u$. The resulting signal (MU $\curvedrightleft$ SU) occupies a total bandwidth of $B_{tot}=B_d+B_u$, as detailed in Section \ref{sect:FDComm}.
The SU relies solely on the STMM to encode the information, opening up intriguing applications in severe energy-constrained scenarios.
To summarize, the main contributions of the paper are as follows: 
\begin{itemize}
    \item We propose a mathematical model for STMM-based wireless communication (MU $\curvedrightleft$ SU). The SU terminal retro reflects the signal adding information to the MU by time-modulating the phase of the STMM. One of the major innovations proposed in this paper is the use of a modulated feed signal at the STMM for signal MU $\leftarrow$ SU over retro reflections. 
    
    \item We analytically characterize the coupling between the spatial and temporal phases in STMM, evaluating their practical impact on the proposed communication model. Noteworthy, none of the current literature works consider the coupling of temporal and spatial phase gradients. We discuss two possible STMM implementations, one allowing the phase of each meta-atom of the STMM to be changed in both time and space and the other, cost-effective, that only requires a single time-varying component, that provides a temporal phase gradient common to all the meta-atoms. 
    Our results show that space-time coupling strongly limits the reception angles, the size of the metasurface, and the MU $\leftarrow$ SU communication rate, giving rise to a cut-off to retro-reflection MU $\leftarrow$ SU bandwidth limit. Further, we also formalize a method for space-time coupling compensation within the cut-off region, enabling a flexible and reliable implementation of the proposed communication model.

    \item We analytically derive the upper-bound performance of the proposed model to form the guidelines for proper STMM design based on the main design parameters, i.e., transmitted power, MU-SU distance, and size of the metasurface. 
    
    \item We analyze the maximum spectral efficiency achievable of MU $\curvedrightleft$ SU considering the generic phase modulation for retro-reflection, namely continuous phase modulation (CPM) with and without countermeasures for space-time coupling. We also analytically and numerically evaluate the impact of channel estimation errors on the MU $\leftarrow$ SU spectral efficiency, affecting the spatial phase configuration of the STMM and related space-time phase decoupling.
\end{itemize}    

\textit{Organization}: The remainder of the paper is organized as follows: Section \ref{sect:AnomalousNonreciprocal} outlines the fundamental equations describing the STMM behavior, system and channel model are presented in Section \ref{sect:SysModel}, Section   \ref{sec:STCDesign} discusses the design principles and challenges of the STMM, while UL signal design is presented in Section \ref{sect:FDComm}. Finally, numerical results and conclusions are in Section \ref{sect:numerical_results} and Section \ref{sect:conclusion}, respectively.

\textit{Notation}: Bold upper- and lower-case letters describe matrices and column vectors. The $(i,j)$-th entry of matrix $\mathbf{A}$ is denoted by $[\mathbf{A}]_{(i,j)}$. Matrix transposition, conjugation, conjugate transposition, and Frobenius norm are indicated respectively as $\mathbf{A}^{\mathrm{T}}$, $\mathbf{A}^{*}$, $\mathbf{A}^{\mathrm{H}}$ and $\|\mathbf{A}\|_F$. $\otimes$ denotes the Kronecker (tensor) froduct between matrices, $\mathrm{tr}\left(\mathbf{A}\right)$ extracts the trace of $\mathbf{A}$. $\mathrm{diag}(\mathbf{A})$ denotes the extraction of the diagonal of $\mathbf{A}$, while $\mathrm{diag}(\mathbf{a})$ is the diagonal matrix given by vector $\mathbf{a}$. $\mathbf{I}_n$ is the identity matrix of size $n$. With  $\mathbf{a}\sim\mathcal{CN}(\boldsymbol{\mu},\mathbf{C})$ we denote a multi-variate circularly complex Gaussian random variable $\mathbf{a}$ with mean $\boldsymbol{\mu}$ and covariance $\mathbf{C}$. $\mathbb{E}[\cdot]$ is the expectation operator, while $\mathbb{R}$ and $\mathbb{C}$ stand for the set of real and complex numbers, respectively. $\delta_{n}$ is the Kronecker delta.

\section{Anomalous Nonreciprocal Metasurface}\label{sect:AnomalousNonreciprocal}

\begin{figure}[t]
    \centering
    \includegraphics[width=0.6\columnwidth]{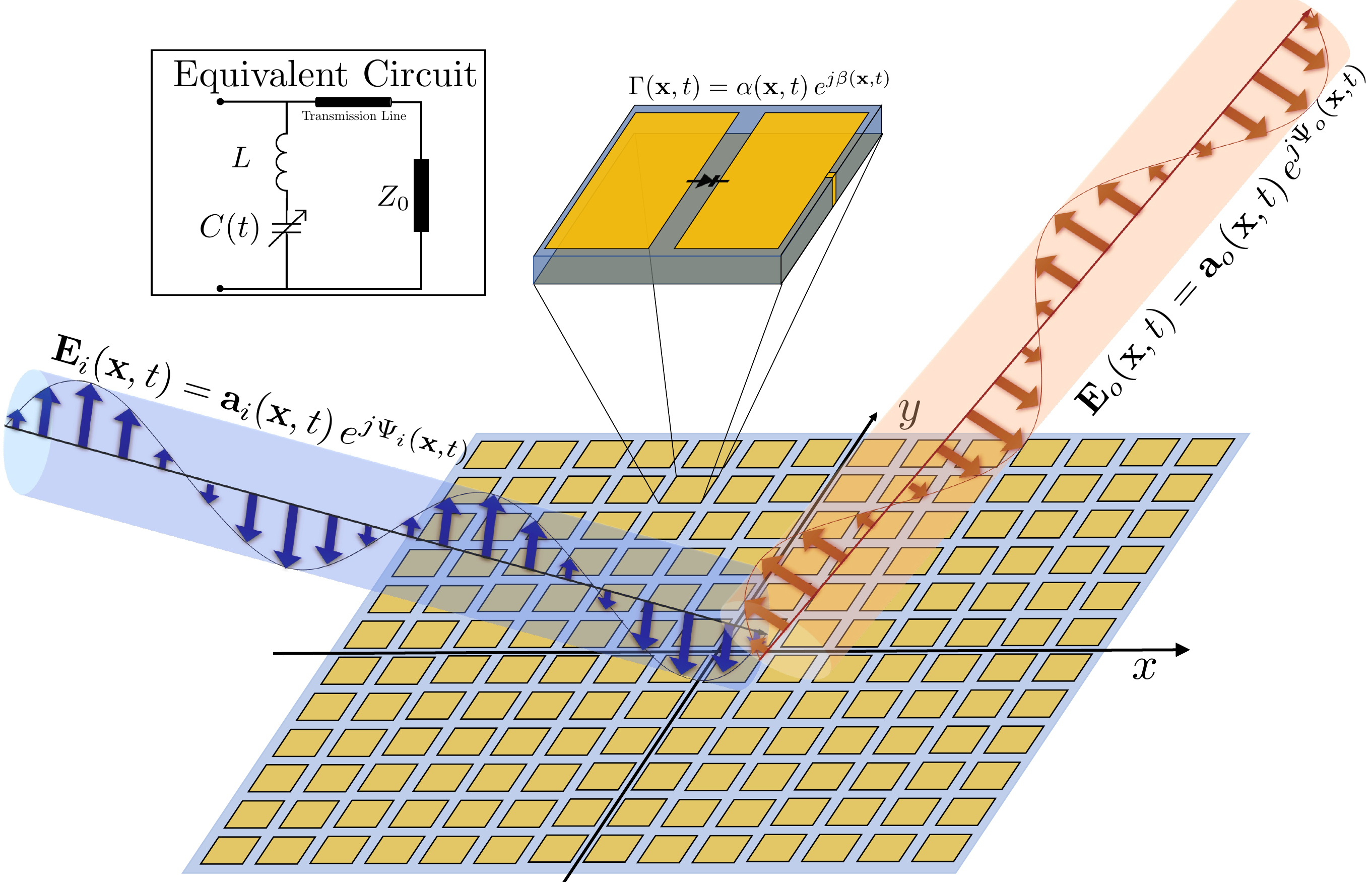}
    \caption{2D STMM Scheme.}
    \label{fig:refSystem}
\end{figure}

This section outlines the fundamental equations describing the behavior of an STMM. The universal Snell's law is a generalization when dealing with both spatial and temporal phase gradients across a metasurface. 
Let us consider the scheme depicted in Fig. \ref{fig:refSystem}, where a plane EM wave $\mathbf{E}_i$ is impinging on a 2D planar metasurface deployed on the $xy$ plane. The incident wave can be expressed as 
\begin{equation}\label{eq:IncidentWave1}
    \mathbf{E}_i(\mathbf{x},t) = \mathbf{a}_i \, e^{j\Psi_i(\mathbf{x},t)},
\end{equation}
where $\mathbf{a}_i$ is the amplitude of the wave, whereas phase $\Psi_i(\mathbf{x},t)$ is fast varying in space and time. 
Quantities 
\begin{align}
        \frac{1}{2\pi}\pdv{t}{\Psi_i(\mathbf{x},t)} & = f_i\label{eq:fi}\\
        \nabla_\mathbf{x} \Psi_i(\mathbf{x},t) & = \mathbf{k}_i,\label{eq:Ki}.
\end{align}
are, respectively, the frequency $f_i$ and the wavevector $\mathbf{k}_i$ of the incident wave, whose expression in the adopted reference system (Fig. \ref{fig:refScen}) is 
\begin{equation}\label{eq:inc_K}
    \mathbf{k}_i = \frac{2 \pi}{\lambda_i} \left[\cos\theta_i \cos\phi_i, \cos\theta_i \sin\phi_i, \sin\theta_i\right]^T
\end{equation}
for wavelength $\lambda_i=c/f_i$ ($c$ is the wave speed).
Let us define the reflection coefficient of the metasurface as
\begin{equation}\label{eq:rm}
    \Gamma(\mathbf{x},t) = \alpha(\mathbf{x},t) \, e^{j\beta(\mathbf{x},t)},
\end{equation}
where $\alpha(\mathbf{x},t)$ and $\beta(\mathbf{x},t)$ denote the time- and space-varying amplitude and phase of the reflection coefficient of the STMM, respectively. Herein, the metasurface is a phase-only device, thus $\alpha(\mathbf{x},t)=\alpha\neq 0$ on the metasurface area. Similarly to \eqref{eq:IncidentWave1}, the reflected wave can be generally expressed as
\begin{equation}\label{eq:ReflectedWave}
    \mathbf{E}_o(\mathbf{x},t) = \mathbf{a}_o \, e^{j\Psi_o(\mathbf{x},t)},
\end{equation}
with $\mathbf{a}_o\approx \mathbf{a}_i$ under the assumption of ideal reflection, i.e., $\alpha = 1 \, \forall \mathbf{x},\, t$. The reflected wave is characterized by wavevector $\mathbf{k}_o$ and frequency $f_o$, that follow from the universal Snell's law as \cite{IJ133_NatComm_9_4334_2018}:
\begin{align}
    \mathbf{k}_o - \mathbf{k}_i &= \left[\pdv{x}{\beta(\mathbf{x},t)},\pdv{y}{\beta(\mathbf{x},t)}, 0\right]^\mathrm{T}, \label{eq:spaceGradient}\\
    f_o - f_i &= \frac{1}{2\pi} \pdv{t}{\beta(\mathbf{x},t)}.\label{eq:timeGradient}
\end{align}
As it can be observed in \eqref{eq:timeGradient}, any spatial variation of the phase $\beta(\mathbf{x},t)$ produces a variation of the instantaneous wavevector $\mathbf{k}_o$, and a time variation of the phase $\beta(\mathbf{x},t)$ induces a variation of the instantaneous frequency $f_o$ in the reflected wave.
The design of $\beta(\mathbf{x},t)$ to control both direction $(\mathbf{k}_o)$ and frequency $(f_o)$ of the reflected wave is discussed in Section \ref{sec:STCDesign}.

\section{System and Channel Model}\label{sect:SysModel}

\begin{figure*}[t!]
    \centering
    \includegraphics[width=0.98\textwidth]{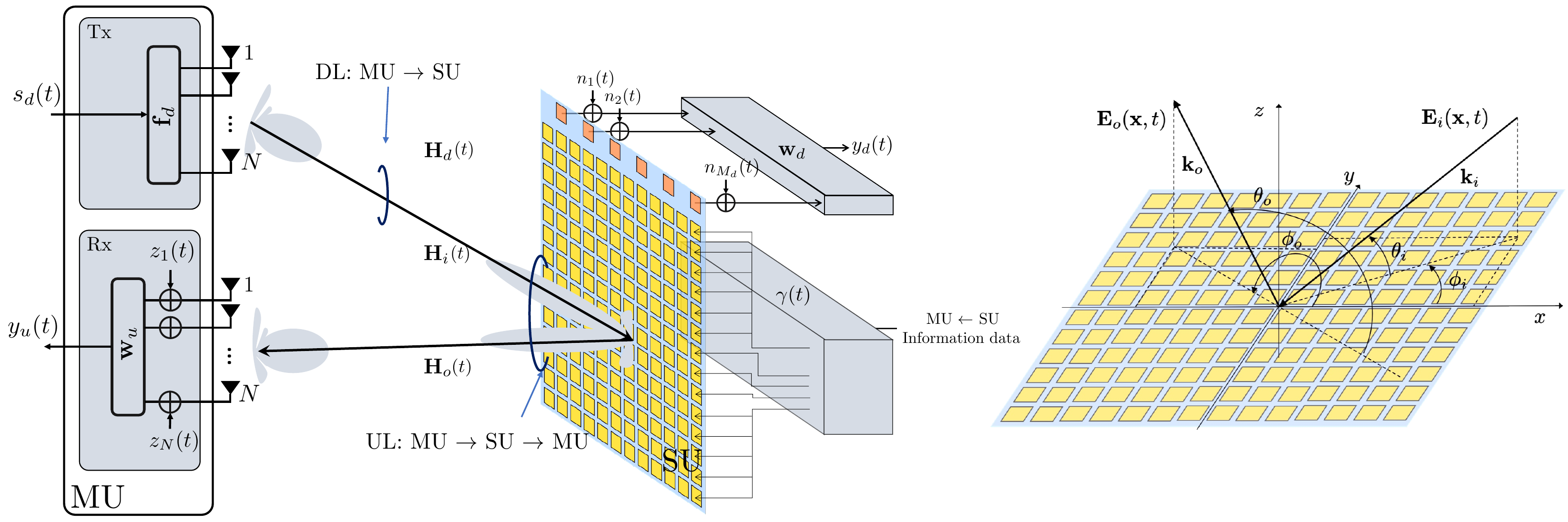}
    \caption{Proposed system model}
    \label{fig:refScen}
\end{figure*}
Let us consider the reference scenario depicted in Fig. \ref{fig:refScen}, where a generic MU establishes a full-duplex link with a SU. To simplify, the MU is equipped with two uncoupled antenna arrays of $N$ meta-atoms each, to enable full-duplex capabilities at MU. The SU is equipped with a planar STMM of $M_u=M_{u,x}\times M_{u,y}$ meta-atoms (along $x$ and $y$ respectively), and an Rx antenna array of $M_d$ meta-atoms. We assume that the STMM and the Rx array are co-located at the SU. The assumption concerning the SU is meant to conceptually separate for simplicity the device that receives MU $\rightarrow$ SU data from the MU $\leftarrow$ SU data stream via retro-reflection. Note that, it is possible to make a single device capable of both receiving and transmitting at the same time, for example, by connecting some meta-atoms of the metasurface to one (or more) RF chain \cite{Schroeder2022_hybridRIS}. In this case, the analytical derivations would be slightly different, but the results are equivalent to those investigated herein.

\subsection{Signal model}
The time-continuous signal from the MU can be generally expressed as: 
\begin{equation}\label{eq:Tx_signal_w_precoding}
    \mathbf{x}(t) = \mathbf{f}_d\;s_d(t)e^{j2\pi f_i t}
\end{equation}
where $\mathbf{f}_d\in\mathbb{C}^{N\times 1}$ is the spatial precoding vector and $s_d(t)$ is the MU $\rightarrow$ SU time-continuous base-band information signal with bandwidth $B_d$, centered around the carrier frequency $f_i$. 
The MU $\rightarrow$ SU signal propagates over a high-frequency channel \cite{rappaport2019wireless}, yielding a received signal at the SU:
\begin{equation}\label{eq:dlSignal}
\begin{split}
    y_d(t) & = \frac{1}{\sqrt{\varrho_d}} \, \mathbf{w}_d^H \mathbf{H}_d(t)* \mathbf{x}(t) + \mathbf{w}_d^H \mathbf{n}(t) 
\end{split}
\end{equation}
where $\varrho_d$ represents the MU $\rightarrow$ SU path loss in power, * is the matrix convolution,  $\mathbf{w}_d\in\mathbb{C}^{N\times 1}$ is the SU spatial combiner, $\mathbf{H}_d(t)\in \mathbb{C}^{M_d \times N}$ is the MU $\rightarrow$ SU MIMO channel matrix, such that $\mathrm{E}\left[\|\mathbf{H}_d(t)\|_\mathrm{F}^2\right] = N M_d$ while  $\mathbf{n}(t)\sim\mathcal{CN}(\mathbf{0},\sigma^2_{n}\mathbf{I}_{M_d} \delta(t))$ is the additive noise, spatially and temporally uncorrelated.

The Rx signal model at the MU can be evaluated by considering the propagation towards the STMM and its retro-reflection to the MU. The MU $\rightarrow$ SU signal at the STMM is
\begin{equation}\label{eq:STCMSignal}
    \mathbf{r}(t) = \mathbf{H}_i(t) * \mathbf{x}(t) 
\end{equation}
where $\mathbf{H}_i(t)\in\mathbb{C}^{M_u\times N}$ is the \textit{forward} MIMO channel between the MU and the STMM. The signal $\mathbf{r}(t)$ undergoes the STMM phase reflection and then the convolution with the \textit{backward} (MU $\leftarrow$ SU) MIMO channel $\mathbf{H}_o(t)\in\mathbb{C}^{N\times M_u}$, obtaining:
\begin{equation}\label{eq:ulSignal}
\begin{split}
       y_u(t) & = \frac{1}{\sqrt{\varrho_u}} \mathbf{w}_u^H \mathbf{H}_o(t) * \boldsymbol{\Gamma}(t)\mathbf{r}(t) + \mathbf{w}_u^H \mathbf{z}(t) = \\  
       & = \frac{1}{\sqrt{\varrho_u}} \mathbf{w}_u^H \mathbf{H}_o(t) * \boldsymbol{\Gamma}(t)\left(\mathbf{H}_i(t) * \mathbf{x}(t) \right) +\\
       &+ \mathbf{w}_u^H \mathbf{z}(t),
\end{split}
\end{equation}
where $\varrho_u$ is the UL path loss in power, $\mathbf{w}_u\in\mathbb{C}^{N\times 1}$ is the uplink combiner at the MU, $\mathbf{z}(t)\sim\mathcal{CN}(\mathbf{0},\sigma^2_{z}\mathbf{I}_{N} \delta(t))$ is the additive noise at the MU array and $\boldsymbol{\Gamma}(t) \in \mathbb{C}^{M_u \times M_u}$ denotes the time-varying STMM phase pattern matrix
\begin{equation}\label{eq:refMtx}
    \mathbf{\Gamma}(t) = \mathrm{diag}\left(\left[e^{j\beta_{0,0}(t)}, \cdots,  e^{j\beta_{M_{u,x}-1, M_{u,y}-1}(t)}\right]\right).
\end{equation}
%
where the time-varying phase matrix is obtained by discretizing the reflection coefficient in \eqref{eq:rm} over the spatial dimension, i.e., the STMM is modeled as a modulated reflectarray with regularly spaced meta-atoms. In particular, the relative position of the $(q,v)$-th meta-atom of the STMM w.r.t. to the phase center is $\mathbf{p}_{q,v} = [q\,d_x, v\,d_y,0]^T$, for $q = 0,..., M_{u,x}-1$, $v=0,..., M_{u,y}-1$, being $d_x$ and $d_y$ being the meta-atoms spacing along the $x$ and $y$ axes, respectively. Notice that, although $\mathbf{\Gamma}(t)$ is a time-varying function, the relation with $\mathbf{r}(t)$ is multiplicative because the time-varying signal injected by the STMM is a reflection coefficient.

\textit{Remark 1}: Precoding and combining vectors, i.e., $\mathbf{f}_d$ and ($\mathbf{w}_d$, $\mathbf{w}_u$), are herein based on the knowledge of channels $\mathbf{H}_d(t)$ and $\mathbf{H}_i(t)$, $\mathbf{H}_o(t)$ at both MU and SU side. Notice that channel estimation at the SU can be obtained through conventional approaches and shared with the MU through feedback~\cite{channelMiz}.

\textit{Remark 2}: Decoding of the MU $\leftarrow$ SU information signal in $\beta_{q,v}(t)$ is possible thanks to the knowledge of the MU $\rightarrow$ SU one $s_d(t)$, that shall be removed by the MU to obtain $\beta_{q,v}(t)$ (not analyzed here). 

\subsection{Channel Model}\label{subsec:chModel}

The cluster-based multipath channel model is widely used in the design of mmWave and sub-THz wireless communication systems because it effectively captures the unique propagation characteristics of these frequencies \cite{Meijerink2014SalehValenzuela}. This model is used here since this paper is framed for these frequencies. The $(m,n)$-th entry of any of the channel impulse responses in \eqref{eq:dlSignal}-\eqref{eq:ulSignal} can be written as
\begin{align}
    [\mathbf{H}(t)]_{m,n} & = \sum_{p=1}^P \xi_p \,\delta\left(t-\tau_p-\Delta t_{p,n}-\Delta t_{p,m}\right) \label{eq:inChannel}
\end{align}
where \textit{(i)} $P$ is the number of paths, \textit{(ii)} $\xi_p\sim \mathcal{CN}\left(0, \sigma_p^2\right)$ is the scattering amplitude of $p$-th path, such that $\sum_{p=1}^P \sigma_p^2 = 1$, \textit{(iii)} $\tau_p$ is the propagation delay of the $p$-th path from the phase center of the Tx to the phase center of the Rx, \textit{(iv)} $\Delta t_{p,n}$ and $\Delta t_{p,m}$ are the residual propagation delays due to the position of the $n$-th Tx meta-atom and the $m$-th Rx meta-atom w.r.t. to their phase centers.
As common for communication systems operating through back-scattering, the two-way channel to/from the STMM is characterized by a single dominant path, as for radar systems \cite{manzoni2022motion}. However, the technical extent of this work is not limited by the latter assumption. Multiple bounces (i.e., multipath forward $\mathbf{H}_i(t)$ and backward $\mathbf{H}_o(t)$ channels) can be considered, but \textit{(i)} their path loss is usually higher (or much higher) w.r.t. the direct MU $\curvedrightleft$ SU path, except for pure reflections in the environment (e.g., walls) and \textit{(ii)} both MU and SU are spatially selective, thus multipath components outside the main radiation lobe of the MU (or SU) would be strongly attenuated. The combination of spatial selectivity and increased path-loss allows us to reasonably assume $P=1$ for both forward $\mathbf{H}_i(t)$ and backward $\mathbf{H}_o(t)$ channels. Considering a single path yields a common propagation delay $\tau_p=\tau=D/c$ between MU and SU  and residual delays at the STMM that depend on the incidence and reflection angles $\theta = \theta_o = \theta_i$, $\phi= \phi_o = \phi_i$ (coincident for full-duplex settings such as the one depicted in Fig. \ref{fig:refScen}). The expression of the residual delay $\Delta t$ at the STMM is shown in Section \ref{subsect:principles}. The inter-meta-atom spacing of Tx and Rx arrays (at both MU and SU) is $\lambda_i/2$, while at STMM is $d_x=d_y=\lambda_i/4$. 
The propagation path losses in power, $\varrho_d$ and $\varrho_u$ in \eqref{eq:dlSignal}-\eqref{eq:ulSignal} are modeled as \cite{Ellingson2021}
\begin{equation}\label{eq:pathloss}
    \varrho_d = \frac{2^4 \pi D^2}{\lambda_i^2}, \quad \quad \varrho_u = \frac{2^{12} \pi D^4}{\lambda_i^4}
\end{equation}
where $D$ is the MU-SU distance. The model in \eqref{eq:pathloss} is obtained by assuming that the effective area of the MU/SU arrays is equal to the physical one, which is approximated as $N(\lambda_i/2)^2$, and that the retro-reflected cross-section of the STMM is equal to the one of a perfect metallic plate with an area of $M_u (\lambda_i/4)^2$ as if it is oriented towards the MU \cite{Ellingson2021}. 

\textit{Remark}: The MU $\leftarrow$ SU path loss $\varrho_u$ is here assumed as only dependent on the wavelength of the incident signal $\lambda_i$. As detailed in Section \ref{sec:STCDesign}, this is true for STMM phases $\beta_{q,v}(t)$ considered as stochastic, whose average time derivative is zero, i.e., $\mathbb{E}[\mathrm{d}\beta_{q,v}(t)/\mathrm{d}t] = 0 $, as common for information-bearing signals. Differently, for $\mathrm{d}\beta_{q,v}(t)/\mathrm{d}t \neq 0$, (e.g., a simple frequency shifting), the denominator of $\varrho_u$ would be $\lambda_i^2 \lambda_o^2$, with $\lambda_o$ derived from \eqref{eq:timeGradient}.

\section{Temporal Modulation}\label{sec:STCDesign}

\begin{figure}[t]
    \centering
    \subfloat[Architecture A]{\includegraphics[width=0.45\textwidth]{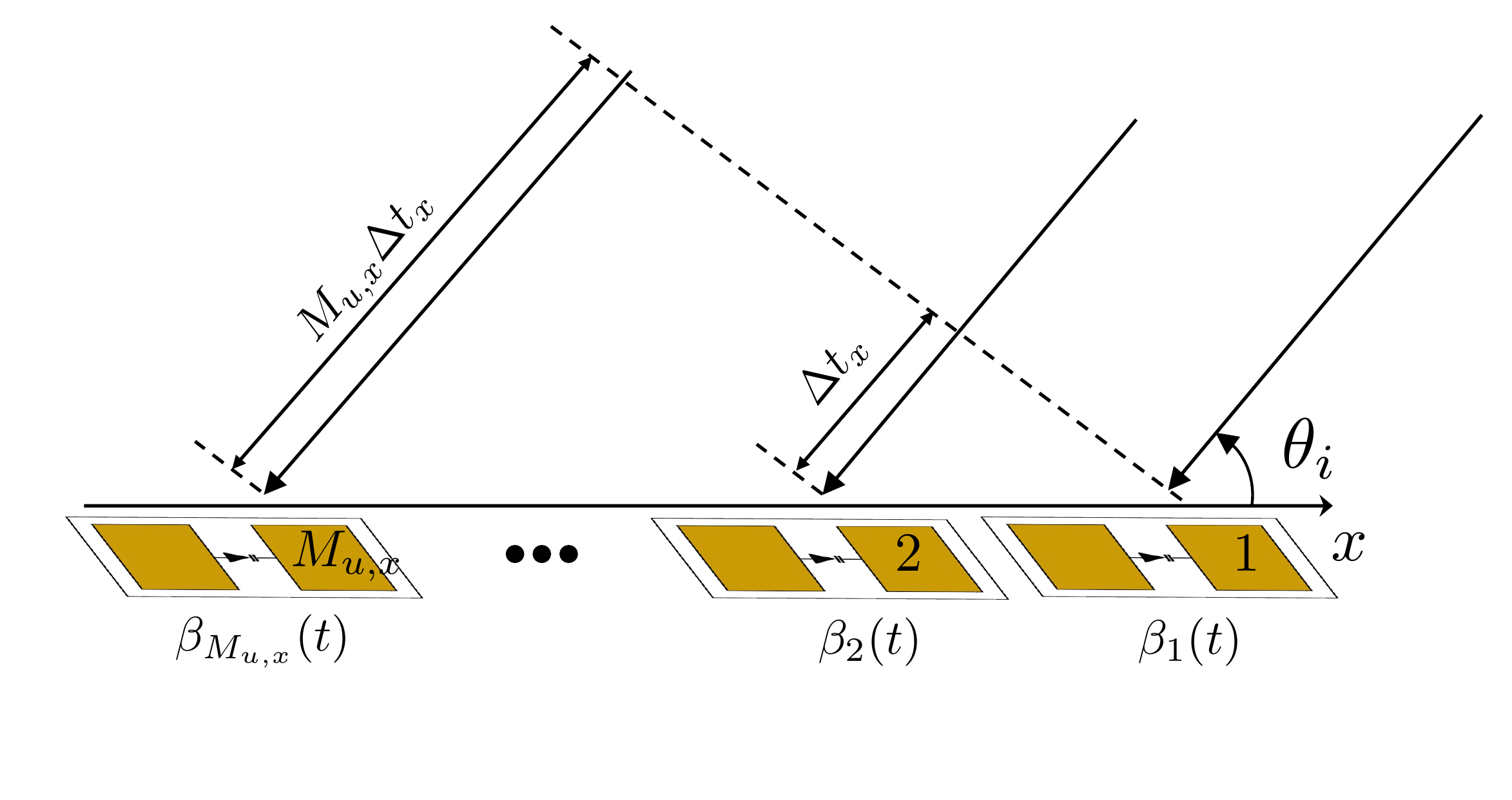}\label{subfig:ArchA}}\hspace{0.5cm}
    \subfloat[Architecture B]{\includegraphics[width=0.45\textwidth]{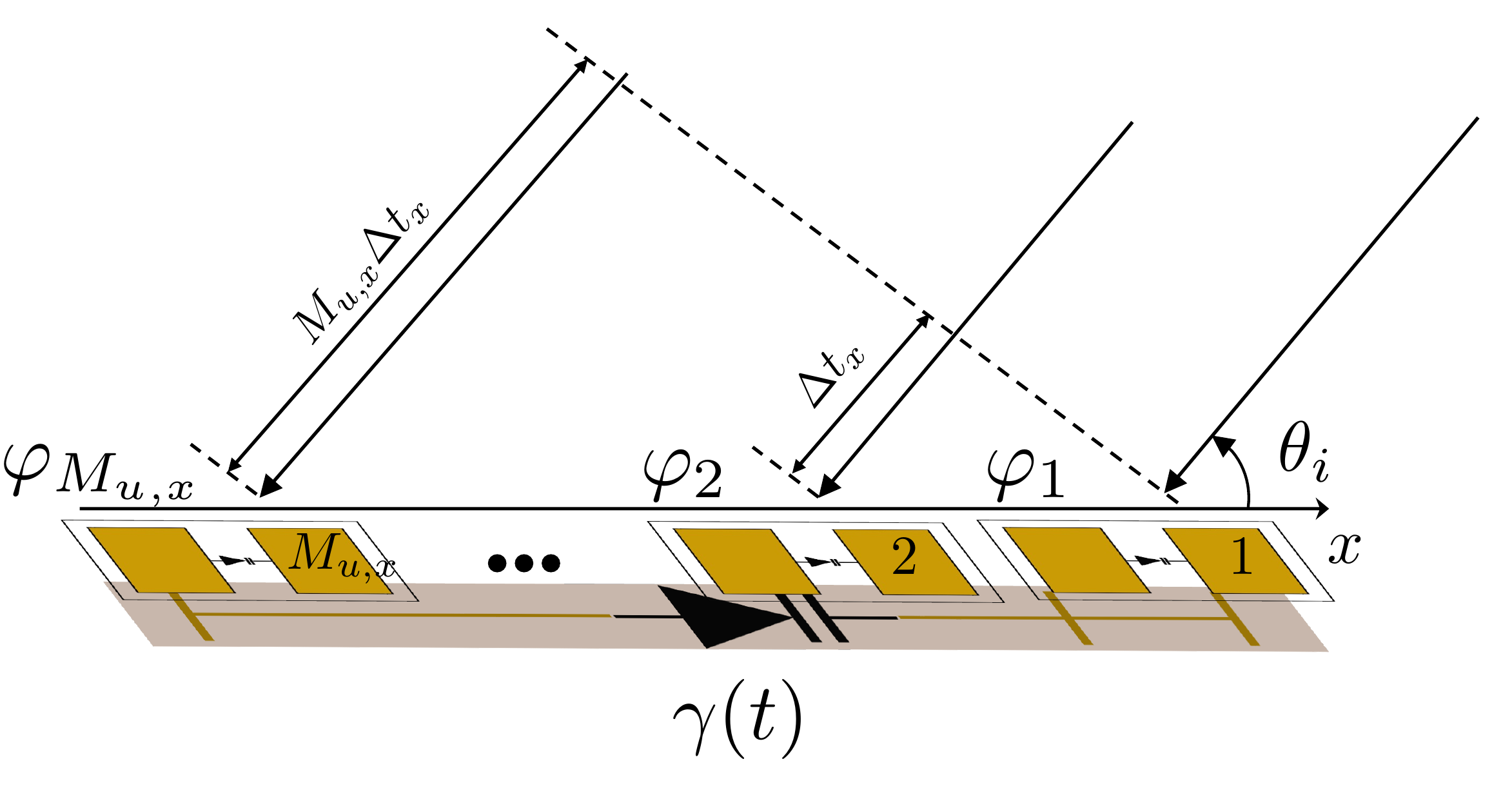}\label{subfig:ArchB}}
    \caption{Representation of a 1D plane wave impinging on the STMM for (\ref{subfig:ArchA}) Architecture A, considering each meta-atom with tunable space-time capabilities, and  (\ref{subfig:ArchB}) Architecture B, meta-atoms spatially tunable and a single tunable component shared across the whole metasurface.}
    \label{fig:SPhaseDist}
\end{figure}

This section defines the principles of temporal modulation at the STMM and its peculiarities. For the implementation of the STMM, we propose two potential architectures, as depicted in Fig. \ref{fig:SPhaseDist}. The first architecture, A, models each meta-atom with space-time tunability. Alternatively, architecture B models each meta-atom with spatial tunability and a single tunable component for temporal control shared with all meta-atoms. The difference between A and B is mainly the reduced cost/complexity of the second, which comes at the price of a slightly reduced performance, as will be discussed in the following. This dissimilarity arises from the fact that temporal control requires fast-varying varactors to support very high data rates in the MU $\leftarrow$ SU link, whereas spatial control can be achieved using cost-effective pin diodes. This discrepancy in component selection is attributed to the distinct operational demands of each control mechanism. To ease the derivation while maintaining the generality, we consider the MU to use precoding and combining vectors perfectly aligned with the SU. Thus, we have that $\mathbf{f}_d$ and $\mathbf{w}_u$ are given, e.g., estimated by conventional approaches~\cite{channelMiz}. 

\subsection{Principles of Temporal Modulation}\label{subsect:principles}
Let us consider the phase \eqref{eq:refMtx} applied at the STMM meta-atoms (architecture B) as the superposition of a spatial and a temporal component, such that at the $(q,v)$-th meta-atom is
\begin{equation}\label{eq:STPhase}
    \beta_{q,v}(t) = \varphi_{q,v} + \gamma(t),
\end{equation}
for $q=0,..., M_{u,x}-1$, $v=0,..., M_{u,y}-1$, where $\varphi_{q,v}$ is the space component of the phase used to design a complete retro-reflection of the impinging signal and $\gamma(t)$ is an information-bearing phase signal applied at each STMM meta-atom, and common to all. Model \eqref{eq:STPhase}, adopted in the current literature \cite{9133266, 9695952} can also be considered valid for architecture A when only one data stream is generated at the SU, as considered in this paper. 

Expanding \eqref{eq:ulSignal} with \eqref{eq:STPhase}, we obtain \eqref{eq:rxSignalExpanded}
\begin{figure*}[h!]
\begin{equation}\label{eq:rxSignalExpanded}
\begin{split}
    y_u(t) & = \rho  \sum_{q=0}^{M_{u,x}-1} \sum_{v=0}^{M_{u,y}-1} e^{-j4 \pi f_i \Delta t_{q,v}} e^{j\varphi_{q,v}} e^{j\gamma(t-\Delta t_{q,v}-\tau)} s_d(t - 2\Delta t_{q,v}-2\tau) + z(t) \overset{(a)}{\approx} \\
    & \overset{(a)}{\approx} \rho \, s_d(t-2\tau) \underbrace{\sum_{q=0}^{M_{u,x}-1} \sum_{v=0}^{M_{u,y}-1} e^{-j4 \pi f_i \Delta t_{q,v}} e^{j\varphi_{q,v}} e^{j\gamma(t-\Delta t_{q,v}-\tau)}}_{\text{Multiplicative MU$\leftarrow$SU channel}\, h_u(t)}  + z(t) \overset{(b)}{\approx}\\
    & \overset{(b)}{\approx} \rho \, s_d(t-2\tau) \sum_{q=0}^{M_{u,x}-1} \sum_{v=0}^{M_{u,y}-1} e^{-j4 \pi f_i \Delta t_{q,v}} e^{j\varphi_{q,v}} e^{j\gamma(t-\tau)} + z(t)
\end{split}
\end{equation}
\hrulefill
\end{figure*}
where $\rho= N^2 \xi_i\xi_o/\sqrt{\varrho_u}$ accounts for geometrical energy losses, while
\begin{equation}\label{eq:propagationDelay_2D}
    \Delta t_{q,v} = \frac{1}{c} \left(\frac{\mathbf{p}_{q,v}^T\mathbf{k}_i}{\|\mathbf{k}_i\|}\right)=  q \Delta t_x + v \Delta t_y
\end{equation}
is the residual propagation delay for the $(q,v)$-th STMM meta-atom, located in $\mathbf{p}_{q,v}= (\lambda_i/4) [q, v, 0]^T$ (local STMM coordinates), and 
\begin{equation}\label{eq:propagationDelay}
    \Delta t_x = \frac{\lambda_i}{4} \frac{\cos \theta \cos \phi}{c},\,\,\,\,\, \Delta t_y = \frac{\lambda_i}{4} \frac{\cos \theta \sin \phi}{c}
\end{equation}
are the propagation delays of the wavefront across two adjacent STMM meta-atoms displaced by $\lambda_i/4$ along $x$ and $y$. In \eqref{eq:rxSignalExpanded}, the first term is the complete expression of the Rx signal for MU $\rightarrow$ SU $s_d(t)$ and MU $\leftarrow$ SU $\gamma(t)$ signals. Approximation $(a)$ implies that there are no spatially wideband effects at the STMM due to the MU $\rightarrow$ SU signal, namely $s_d(t- 2\Delta T) \approx s_d(t)$, where $\Delta T = M_{u,x} \Delta t_x + M_{u,y} \Delta t_y$ is the maximum delay across the STMM meta-atoms. This is equivalent to asserting that:
\begin{equation}\label{eq:approxA}
    2\Delta T << T_d,
\end{equation}
where $T_d = 1/B_d$ is the symbol period of the MU $\rightarrow$ SU information signal $s_d(t)$.  In this setting, unwanted beam squinting from spatially wideband effects due to $s_d(t)$ can be avoided (see \cite{9399122} for details). With approximation $(a)$, the MU $\rightarrow$ SU signal undergoes a \textit{multiplicative} channel $h_u(t)$ that encodes the MU $\leftarrow$ SU phase signal $\gamma(t)$, whose structure is discussed in Section \ref{subsect:SGradientDistortion}.
Similarly, approximation $(b)$ subtends that the MU $\leftarrow$ SU phase-only signal does not suffer the relative propagation delay across the STMM, thus $\gamma(t-\Delta T )\approx \gamma(t)$, or, equivalently
\begin{equation}\label{eq:approxB}
     \Delta T \ll T_u,
\end{equation}
where $T_u$ is the symbol period of the MU $\leftarrow$ SU phase signal $\gamma(t)$. Approximation $(b)$ in \eqref{eq:rxSignalExpanded} has the further implication that the spatial phase component $\varphi_{q,v}$ is not affected by the temporal one $\gamma(t)$, i.e., they can be separated and, consequently, the reflection gain is maximized by a phase-only pattern (not a function of time) across the STMM:
\begin{align}\label{eq:SPhase}
    \varphi_{q,v} =  \pi (q \cos \theta \cos \phi + v \cos \theta \sin \phi),
\end{align}
leading to the multiplicative model \cite{9133266}
\begin{equation}\label{eq:rxSignalExpanded3}
    \begin{split}
        y_u(t) \approx \rho \, M_u \,  s_u(t-\tau)\,s_d(t-2\tau)  + z(t),
    \end{split}
\end{equation}
where $s_u(t-\tau)=e^{j\gamma(t-\tau)}$ denotes the MU $\leftarrow$ SU signal injected by the STMM \textit{modulating} the MU signal $s_d(t-2\tau)$.

However, model \eqref{eq:rxSignalExpanded3} that is typically assumed in the existing literature on STMM is valid in restricted settings (e.g., $\theta = \pi/2$ and $\phi = 0$, as in \cite{9133266}). Generally, neither approximation $(a)$ nor approximation $(b)$ hold for large-bandwidth systems and/or comparatively large STMM. Noteworthy, while the beam squinting due to $s_d(t)$ cannot be compensated (except by reducing the MU $\rightarrow$ SU bandwidth $B_d$), the coupling between spatial and temporal phase patterns for $\theta \neq \pi/2$ is peculiar to the STMM when involving wireless systems. The detrimental effect analyzed herein is the coupling between the spatial and temporal phases. Therefore, solutions to space-time phase coupling are detailed in Section \ref{subsect:SGradientDistortion} and proposing countermeasure for wideband MU $\leftarrow$ SU system is in Section \ref{subsect:Decoupling}. 

\subsection{Space-Time Phase Coupling} \label{subsect:SGradientDistortion}

The space-time phase coupling can be readily explained considering Fig. \ref{fig:SPhaseDist}. A plane wave impinging on the $(q,v)$-th meta-atom of the STMM with angle $(\theta,\phi)$, as depicted in Fig. \ref{fig:SPhaseDist}, will be retro-reflected with the following phase: 
\begin{equation}\label{eq:truePhaseShift}
    \widetilde{\beta}_{q,v}(t) = \varphi_{q,v} + \gamma(t - \Delta t_{q,v} ),
\end{equation}
thus the phase $\gamma(t)$ experiences a meta-atom-specific time delay that depends on the incidence angles $(\theta,\phi)$, which induces a spurious, and undesired, spatial phase gradient. 
This phenomenon leads to an unwanted decrease in the reflection gain. In other words, assuming the conventional spatial phase gradient detailed in \eqref{eq:SPhase}, we have that the following theorem holds:
\begin{theorem}\label{th:MD}
Let $h_u(t)$ be the channel defined as in \eqref{eq:rxSignalExpanded} and let the phase be defined as \eqref{eq:STPhase}, with $\varphi_{q,v} = \pi (q \cos \theta \cos \phi + v \cos \theta \sin \phi)$, $u=0,...,M_{u,x}-1$, $v=0,...,M_{u,y}-1$. Then, the following inequality holds:
\begin{equation}\label{eq:theorem1}
    \mathbb{E}_\gamma\left[\lvert h_u(t) \rvert^2 \right] \leq M_u^2.
\end{equation}
\end{theorem}
\begin{proof}
The proof can be readily carried out by inspection of the structure of $h_u(t)$, which yields
\begin{equation}\label{eq:theorem2}
    \mathbb{E}_\gamma\left[\bigg\lvert  \sum_{q=0}^{M_{u,x}-1} \sum_{v=0}^{M_{u,y}-1} e^{j\gamma(t-\Delta t_{q,v}-\tau)}\bigg \rvert^2 \right] \leq M_u^2.
\end{equation}
To have equality in \eqref{eq:theorem2}, the arguments of the complex exponential shall neither be a function of $q$ nor of $v$ (spatial and temporal phase decoupling). Therefore, the sufficient conditions are $\gamma(t) = \overline{\gamma}$ (straightforward condition) or $\Delta t_x = \Delta t_y=0$ ($\theta = \pi/2$). 
\end{proof}

Theorem \ref{th:MD} implies that any phase-encoded modulation signal $\gamma(t)$ leads to a strict inequality for \eqref{eq:theorem2}, except for very low MU $\leftarrow$ SU signal bandwidth, fulfilling \eqref{eq:approxB}, or the case in which the impinging wave is perfectly orthogonal to the STMM plane, i.e., $\theta = \pi/2$. 
To evaluate the space-time phase coupling effect, let us consider a generic non-linear $\gamma(t):[a,b] \longrightarrow \mathbb{R}$, infinitely differentiable for $t_0 \in (a, b)$. We can use the Taylor expansion around $t = t_0$, truncated at the $n$-th term
\begin{equation}\label{eq:GammaTaylor}
\begin{split}
    \gamma(t_0\hspace{-0.1cm}-\Delta t_{q,v}) &\approx \gamma(t_0) + \sum_{i=1}^{n-1} \frac{(-\Delta t_{q,v})^i}{i!}  \gamma^{(i)}(t_0).
\end{split}
\end{equation}
where $\gamma^{(i)}(t_0)$ denotes the $i$-th time derivative of $\gamma(t)$ evaluated in $t=t_0$. Considering only the first order ($n=1$) derivative only, namely
\begin{equation}\label{eq:GammaLinear}
\begin{split}
    \gamma(t_0\hspace{-0.1cm}-\Delta t_{q,v}) \approx \gamma(t_0) -\gamma^{(1)}(t_0) \Delta t_{q,v}
\end{split}
\end{equation}
the second term is a linear phase shift with the meta-atom index that results in the so-called \textit{beam squinting}, i.e., the angular shift of the reflected beam w.r.t. to the desired direction $(\theta,\phi)$. This effect causes a reduction of the reflection gain that can be analytically evaluated from the STMM's array factor that adds up to a possible beam squinting due to a large bandwidth $s_d(t)$ (herein not considered). 
\begin{figure}[b!]
    \centering
    \includegraphics[width=0.45\textwidth]{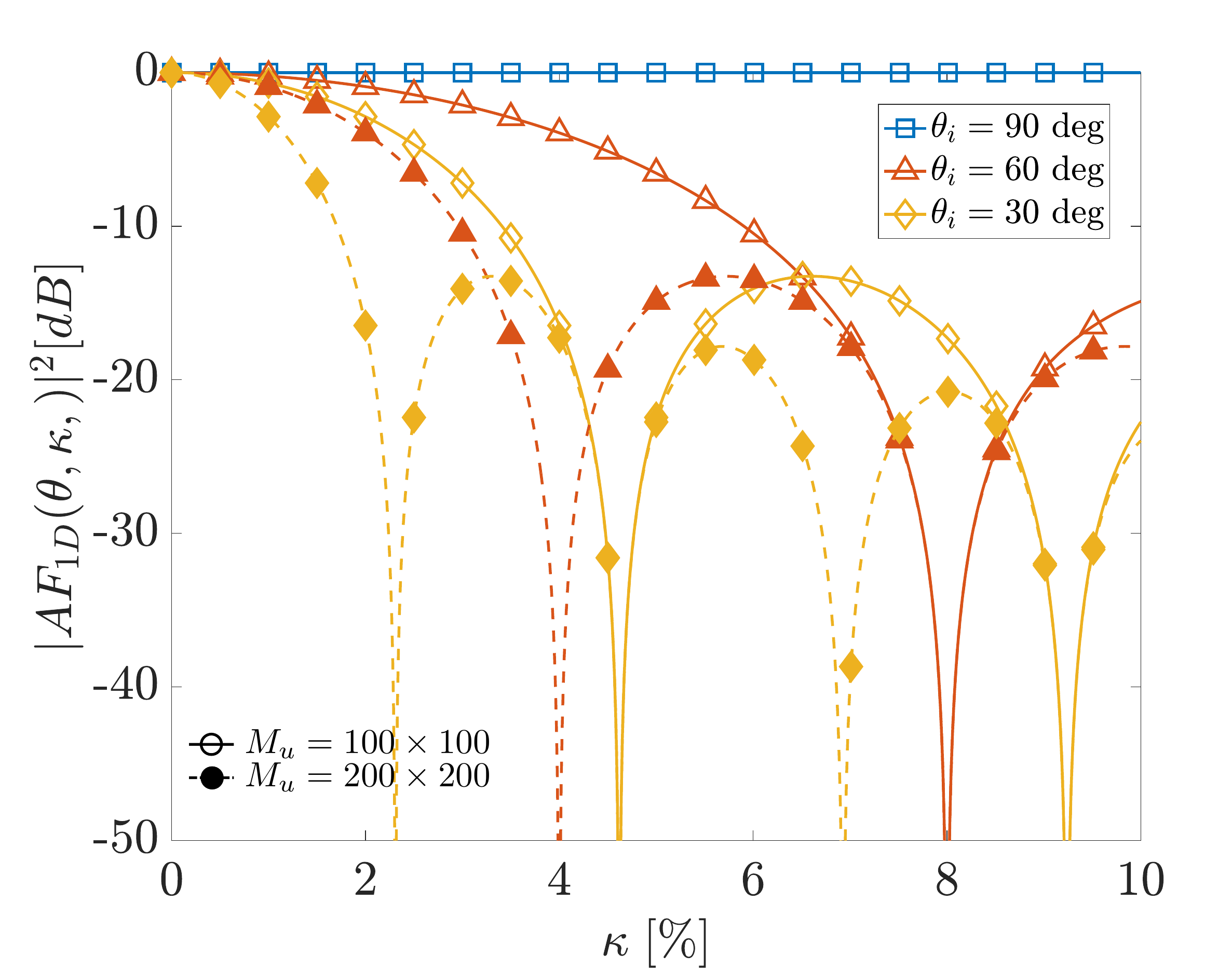}
    \caption{Normalized reflection gain in the presence of space-time phase coupling. The blue line represents the predicted reflection performance by state-of-the-art works for any incident angle (e.g., \cite{9133266}). Due to the coupling, the performance for $\theta_i\neq 90$ deg is in red and purple instead. With the proposed compensation (detailed in Section \ref{subsect:Decoupling}), we can attain the blue line for any value of $\theta_i$.}
    \label{fig:NRG}
\end{figure}

To quantify the beam squinting due to \eqref{eq:GammaLinear}, let us consider a linear temporal phase signal (or piecewise linear, e.g., obtained with frequency shift keying (FSK) modulation) $\gamma(t) = 2\pi f_s t$, where $f_s = f_i - f_o$ is the frequency shift determined by the temporal component of the phase $\beta(\mathbf{x},t)$ as in \eqref{eq:timeGradient}. Let us express the frequency shift as $f_s = \kappa f_i$, where $\kappa \in \mathbb{R}$ is a coefficient that describes the frequency up- or down-conversion, or equivalently $f_o = f_i (1 + \kappa)$. By substituting the expression of $\gamma(t)$ in \eqref{eq:theorem2} and removing the expectation, we obtain
\begin{equation}\label{eq:STCouplingLoss1}
    AF(\theta, \phi, \kappa) = \frac{1}{M_u} \sum_{q=0}^{M_{u,x}-1} \sum_{v=0}^{M_{u,y}-1} e^{j 2 \pi f_i \kappa (t-\Delta t_{q,v})}, 
\end{equation}
with the corresponding normalized reflection gain in \eqref{eq:AF2}.
\begin{figure*}[tb!]
   \centering
\begin{equation}\label{eq:AF2}
    |AF(\theta, \phi, \kappa)|^2 
    = \left|\frac{1}{M_u} \frac{\sin((\pi/4) M_{u,x} \kappa \cos\theta\cos\phi)}{\sin((\pi/4) \kappa \cos\theta\cos\phi)} \frac{\sin((\pi/4) M_{u,y} \kappa \cos\theta\sin\phi)}{\sin((\pi/4) \kappa \cos\theta\sin\phi)}\right|^2.
\end{equation}
    \hrulefill
\end{figure*}
Considering for instance a 1D case (i.e., either a 1D STMM or $\phi=0$), replacing $\Delta t_{q,v}$ in \eqref{eq:propagationDelay_2D} with $\Delta t_{x}$ in \eqref{eq:propagationDelay} ($\phi=0$), the normalized reflection gain simplifies as
\begin{figure}[t!]
    \centering
    \includegraphics[width=0.45\textwidth]{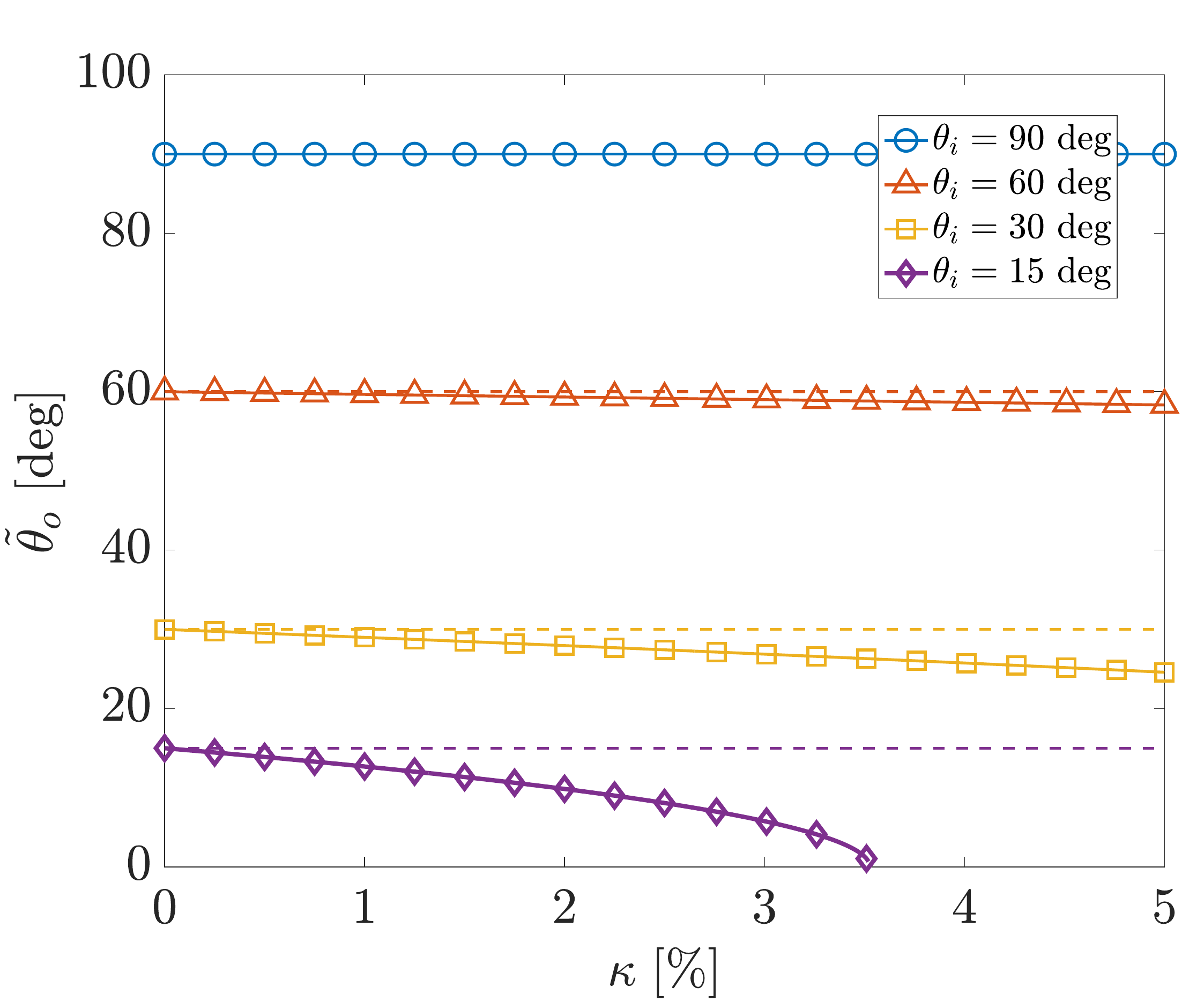}
    \caption{Angle of maximum reflection $\tilde{\theta}_o$, for different incidence angles $\theta_i$ and $M_u=100$.}
    \label{fig:AngDev}
\end{figure}
\begin{equation}\label{eq:AF_drift}
\begin{split}
    |AF_{1D}(\theta, \kappa)|^2
    &= \left| \frac{1}{M_u}\frac{\sin((\pi/2)M_u\kappa \cos\theta)}{\sin((\pi/2)\kappa \cos\theta)}\right|^2.
\end{split}
\end{equation}
This allows gaining insight into the effective reflection gain of the STMM.
Fig. \ref{fig:NRG} depicts the impact of the space-time phase coupling on the normalized reflection gain, and for $\theta \neq \pi/2$ and $\kappa \neq 0$, the system experiences a diminished reflection gain. For instance, at $\theta = \pi/6$ we observe a loss $\geq 3$ dB ($M_u = 100 \times 100$) and $\geq 10$ dB ($M_u = 200 \times 200$) for $\kappa > 2$ \% (e.g., for a frequency shift of $f_s > 0.6$ GHz around $f_i = 30$ GHz). The loss w.r.t. the maximum reflection gain can be justified by the tilting of the direction of maximum reflection, which can be evaluated as
\begin{equation}\label{eq:spatialDrift}
    \tilde{\theta}_o = \arccos \left[(1 + \kappa) \cos \theta \right]
\end{equation}
where $\tilde{\theta}_o$ is the angle of maximum reflection, different from the true one $\theta=\theta_o=\theta_i$ due to the space-time phase gradient coupling. The trend of $\tilde{\theta}_o$ with the frequency shift $\kappa$ is shown in Fig. \ref{fig:AngDev}. The result in \eqref{eq:spatialDrift} is obtained by substituting \eqref{eq:GammaLinear} and \eqref{eq:SPhase} in \eqref{eq:rxSignalExpanded}. Note that for $|\cos \theta (1 + \kappa) | > 1$ in \eqref{eq:spatialDrift} we have evanescent waves, as depicted in Fig.\ref{fig:AngDev2D}.

\begin{figure}[t!]
    \centering
    \includegraphics[width=0.45\textwidth]{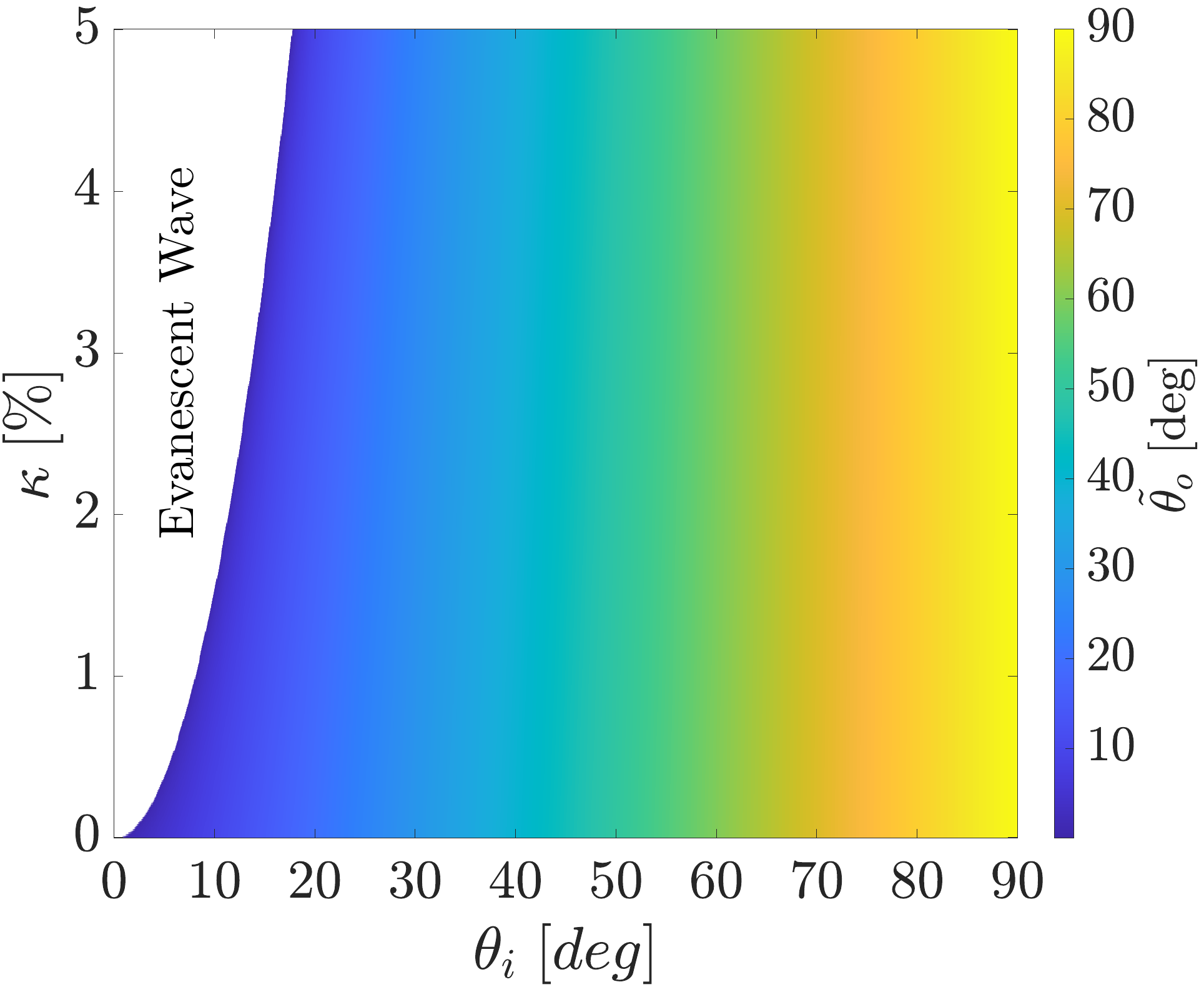}
    \caption{Angle of maximum reflection $\Tilde{\theta}_o$ vs frequency shift $\kappa$ and incidence angle $\theta_i$ ($M_u=100$). For grazing angles, depending on $\kappa$, evanescent waves are excited along the surface.}
    \label{fig:AngDev2D}
\end{figure}

Whenever the phase $\gamma(t)$ is periodic, or cyclo-stationary, the aforementioned approaches are specialized to the harmonics of the Fourier series decomposition of $\gamma(t)$ as experimentally observed in \cite{Zhang2020,Zhang2021}. To implement a wireless communication system using STMM technology, the spatial phase must be decoupled from the temporal one, such that the former can be used to retro-reflect the impinging signal toward the MU and the latter to convey information. 

\subsection{Space-Time Phase Decoupling}\label{subsect:Decoupling}

Decoupling the spatial phase from the temporal phase can be achieved differently depending on the specific STMM architecture.

\subsubsection{Architecture A}

Decoupling can be achieved by delaying the temporal phase of each meta-atom to compensate for the propagation delay caused by the wavefront. Specifically, the phase of the $(q,v)$-th meta-atom must be set as:
\begin{equation}\label{eq:ADecoupling}
    \beta_{q,v}(t) = \varphi_{q,v} + \gamma(t + \Delta t_{q,v}).
\end{equation}
To achieve perfect decoupling, assumed herein, the STMM meta-atoms must be accurately time-synchronized using the knowledge of the incident angles as input, e.g., after a proper estimation.

\subsubsection{Architecture B}

A straightforward yet effective countermeasure to space-time phase coupling is to pre-compensate the spatial phase by considering the induced effect of the temporal phase. In particular, the phase of the $(q,v)$-th meta-atom is
\begin{equation}\label{eq:BDecoupling}
    \beta_{q,v}(t) = \varphi_{q,v} + \gamma(t) + \gamma^{(1)}(t) \Delta t_{q,v} .
\end{equation}
where the spatial phase $\varphi_{q,v}$ now incorporates the meta-atom-dependent delay $\Delta t_{q,v}$. As demonstrated in Section \ref{sect:numerical_results}, architecture B allows for a cost reduction (only a single tunable component is used to control the temporal phase) compared to architecture A, at the price of performance reduction for large bandwidth UL signals. 
The design of $\gamma(t)$, the requirements, and the achievable performance of the proposed communication system are discussed in the following Section \ref{sect:FDComm}.

\section{UL Signal Design}\label{sect:FDComm}

The MU $\rightarrow$ SU signal $y_d(t)$ can be designed following conventional methods (e.g., linear modulation or orthogonal frequency division multiplexing). In this section, we propose the criteria to design the phase shape $\gamma(t)$ for controlled retro-reflection by the STMM.

\subsection{Continuous Phase Modulation}

Let us consider that $\gamma(t)$ is a CPM signal for MU $\leftarrow$ SU:
\begin{equation}\label{eq:cpm}
    \gamma(t) = 2 \pi h \sum_{n} z_n q(t - nT_u)
\end{equation}
where $h$ is the modulation index, $z_n\in\{2b-1\}$, $b\in\{-(M/2),...,(M/2)-1\}$, is the $n$-th $M$-ary information symbol, $T_u$ is the MU $\leftarrow$ SU symbol duration, and $q(t)$ is the phase pulse with the following properties
\begin{equation}
    q(t) =  \begin{cases}
                0 \quad t \leq 0\\
                \frac{1}{2} \quad t \geq L\,T_u
            \end{cases}
\end{equation}
where $L$ represents the modulation memory. The phase pulse is usually defined as
\begin{equation}
    q(t) = \int_{-\infty}^t p(\tau) d\tau
\end{equation}
with $p(t)$ being the pulse shaping filter (PSF). The reason behind choosing CPM is that it represents the most general phase modulation for $\gamma (t)$. Indeed, depending on the parameters set $M, h, L$, and $p(t)$, the CPM can degenerate in the other phase modulations, such as phase shift keying (PSK), frequency shift keying (FSK), minimum shift keying (MSK), and Gaussian MSK (GMSK) \cite{anderson2013digital}. The same reference can be used as an introduction to CPM.

\subsection{Bandwidth Occupancy}

The bandwidth occupied by the MU $\rightarrow$ SU received signal $y_d(t)$ in \eqref{eq:dlSignal}, assuming a frequency-flat channel, is the bandwidth of the transmitted signal $s_d(t)$, i.e., the support of the power spectral density (PSD) $P_{d}(f) = \mathbb{E}\left[\left|\mathcal{F}\left\{s_d(t)\right\}\right|^2\right]$, where $\mathcal{F}\left\{.\right\}$ denotes the Fourier Transform operator. Under the same channel assumption, the bandwidth occupied by the MU $\leftarrow$ SU signal $y_u(t)$ in \eqref{eq:ulSignal}, is the support of its PSD, defined as
\begin{equation}
\begin{split}\label{eq:psd}
    P_{u}(f) &= \mathbb{E}\left[\left|\mathcal{F}\left\{s_d(t) s_u(t)\right\}\right|^2\right] =\\
    &=\mathbb{E}\left[\left|\mathcal{F}\left\{s_d(t)\right\}\right|^2\right] * \mathbb{E}\left[\left|\mathcal{F}\left\{ s_u(t)\right\}\right|^2\right],
\end{split}
\end{equation}
and thus, according to the properties of the convolution, the total bandwidth of the system is
\begin{equation}
    B_{tot} = B_d + B_u
\end{equation}
with $B_u$ being the support of $\mathbb{E}\left[\left|\mathcal{F}\left\{ s_u(t)\right\}\right|^2\right]$. Notice that the resulting spectrum $P_{u}(f)$ is centered around the carrier frequency $f_i$. This follows from the fact that the information-bearing phase signal $\gamma(t)$ has the time-derivative 
\begin{equation}
    \gamma^{(1)}(t) = 2 \pi h \sum_{n} z_n p(t - nT_u)
\end{equation}
that is zero on the ensemble average, i.e., $\mathbb{E}[\gamma^{(1)}(t)] = 0$, thus the average net frequency shift is zero. For a generic CPM signal, the occupied bandwidth is infinite, but, the $99\%$ of the signal energy is within \cite{kuo2004bandwidth} 
\begin{equation}\label{eq:cmpBandwidth}
    B_u = \frac{h}{T_u} \sqrt{\frac{g (M^2 - 1)}{3L}} + \frac{g}{T_uL} = \frac{\epsilon}{T_u}
\end{equation}
where $\epsilon$ is the time-bandwidth product of the specific CPM implementation. In \eqref{eq:cmpBandwidth}, the term $g$ depends on the PSF $p(t)$ used, e.g., $g=1$ for rectangular and $g=1.5$ for raised cosine. The bandwidth occupation is inversely proportional to the CPM memory $L$. For $M=2$, $g=1$ (rectangular PSF), the bandwidth is $B_u = (h\sqrt{L}+1)/(L T_u)$.

It is worth underlining that the proposed communication system is full-duplex, as MU $\rightarrow$ SU and MU $\leftarrow$ SU communications simultaneously occur on the same spectrum portion $B_{tot}$. However, the spectral efficiency is equivalent to an  FDD system, with $B_d$ for MU $\rightarrow$ SU and $B_u$ for MU $\leftarrow$ SU, as detailed in the following. 

\subsection{Spectral Efficiency}

The performance of the proposed system can be assessed in terms of spectral efficiency (SE), considering both MU $\rightarrow$ SU and MU $\leftarrow$ SU links. In particular, the SE upper bound given by Shannon's unconstrained SE is:
\begin{equation}\label{eq:D+Ucapacity}
\begin{split}
        \eta = \underbrace{(1-\mu) \log_2(1+\Upsilon_d)}_{\eta_d} + \underbrace{\mu \log_2(1+\Upsilon_u)}_{\eta_u}
\end{split}
\end{equation}
where the first term ($\eta_d$) is the SE of the MU $\rightarrow$ SU and the second ($\eta_u$) is the SE of the MU $\leftarrow$ SU. In \eqref{eq:D+Ucapacity}, $\mu = B_u/B_{tot}$ is the fractional bandwidth occupied by the MU $\leftarrow$ SU, and $1-\mu = B_d/B_{tot}$ the MU $\rightarrow$ SU counterpart, while $\Upsilon_d$ and $\Upsilon_u$ are the MU $\rightarrow$ SU and MU $\leftarrow$ SU signal-to-noise ratios (SNRs) at the decision variable.

The MU $\rightarrow$ SU SNR is
\begin{equation}\label{eq:snr_d}
    \Upsilon_d \leq \frac{\sigma_{s,d}^2 N M_d}{\varrho_d \sigma_{z}^2}
\end{equation}
where the equality holds when perfect CSI is available and optimal precoding/combing is employed, both providing the maximum MIMO gain, i.e., $N M_d$. 

The MU $\leftarrow$ SU Rx signal undergoes a matched filtering over $T_u$ with the Tx signal $s_d(t-2\tau)$, which is clearly known at MU. Assuming the delay $\tau$ is known, the SNR can be analytically evaluated as in Appendix \ref{app:SNR}, depending on whether the STMM exhibits a wideband behavior w.r.t. $\gamma(t)$ or not. A suitable upper bound on the SNR is obtained by assuming that the wideband regime involves the same UL symbol, i.e., $\Delta T < T_u$, yielding
\begin{equation}\label{eq:snr_u}
    \Upsilon_u \leq \frac{\sigma_{s,d}^2 N^2 M_u^2}{\varrho_u \sigma^2_{z}}  |AF(\theta, \phi, \kappa)|^2 T_u B_{tot}
\end{equation}
where $|AF(\theta, \phi, \kappa)|^2 \leq 1$ is the loss due to the space-time phase coupling, defined in \eqref{eq:AF2}, and $T_u B_{tot}$ is the processing gain of the matched filter. The latter is motivated by the fact that, while the MU $\leftarrow$ SU noise $z(t)$ occupies the whole bandwidth $B_{tot}$, the signal only occupies a fraction of it $\mu B_{tot}$. Thus, the matched filter makes the MU $\leftarrow$ SU SNR inversely proportional to the fractional bandwidth occupation, $\Upsilon_u \propto 1/\mu$. The equality in \eqref{eq:snr_u} holds under the following conditions: \textit{(i)} perfect CSI available at both MU and SU, \textit{(ii)} known delay $\tau$ at the MU, and \textit{(iii)} constant phase $\gamma(t) \in [0, T_u)$.

The MU $\rightarrow$ SU and MU $\leftarrow$ SU capacity in \eqref{eq:D+Ucapacity} has a global optimum with respect to the unconstrained UL bandwidth allocation $\mu$. From the derivative of $\eta$ w.r.t. $\mu$
\begin{equation}\label{eq:D+Lcapacity_derivative}
    \begin{split}
        \frac{\mathrm{d}\eta}{\mathrm{d}\mu} &= -\log_2(1+\Upsilon_d)  - \frac{\widetilde{\Upsilon}_u \log_2e}{\widetilde{\Upsilon}_u + \mu} + \\
        &+ \log_2\left(1+\frac{\widetilde{\Upsilon}_u}{\mu}\right) = 0
    \end{split}
\end{equation}
whose root provides the optimal duplexing allocation. In \eqref{eq:D+Lcapacity_derivative}, we denote with $\widetilde{\Upsilon}_u$ the MU $\leftarrow$ SU SNR terms that do not depend on the allocation $\mu$. The total capacity is analyzed in Section \ref{sect:numerical_results}.

\section{Numerical Results}\label{sect:numerical_results}

\begin{figure}[b!]
    \centering
    \subfloat[]{\includegraphics[width=0.45\textwidth]{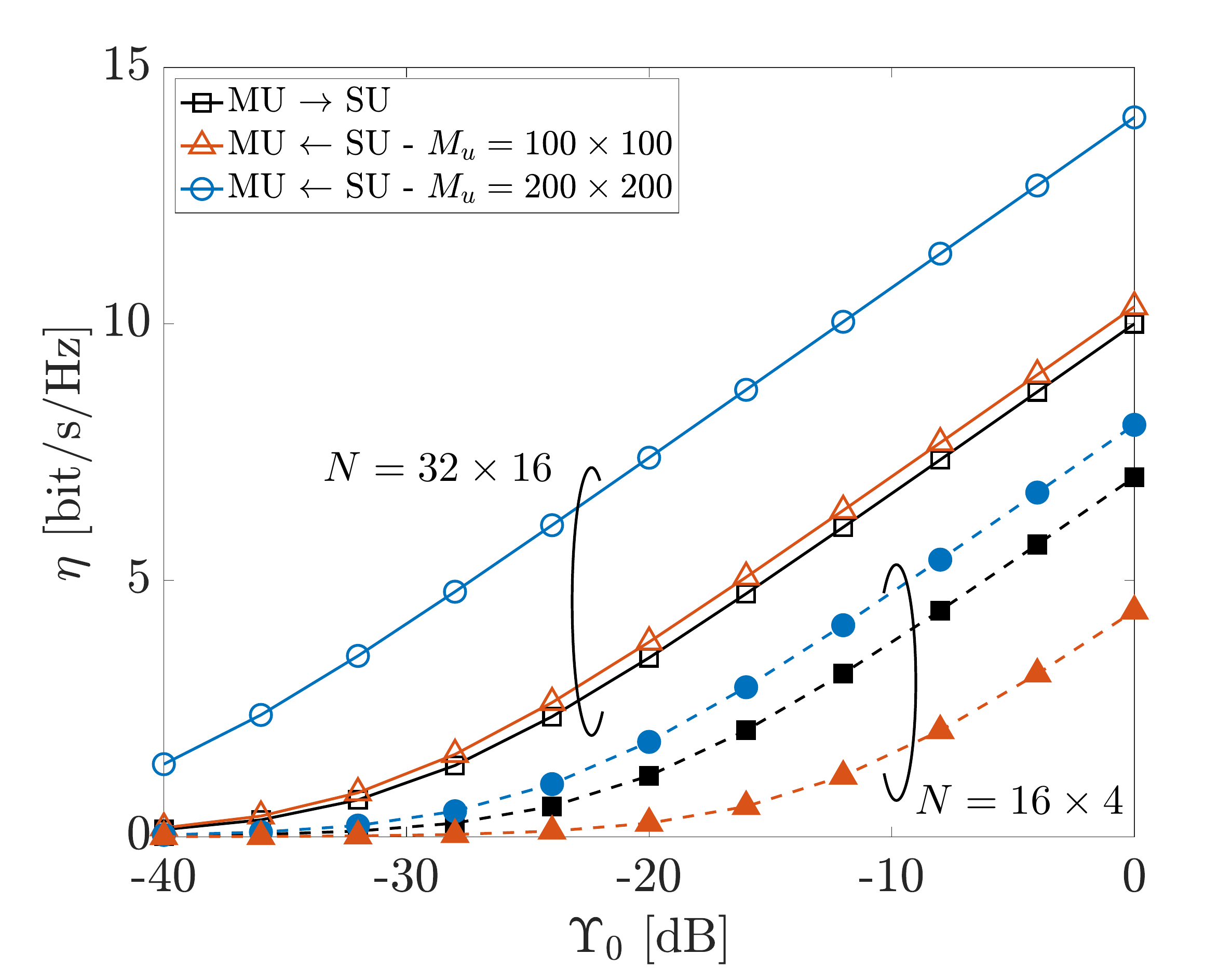} \label{subfig:SE_vs_SNR}}\\
    \subfloat[]{\includegraphics[width=0.45\textwidth]{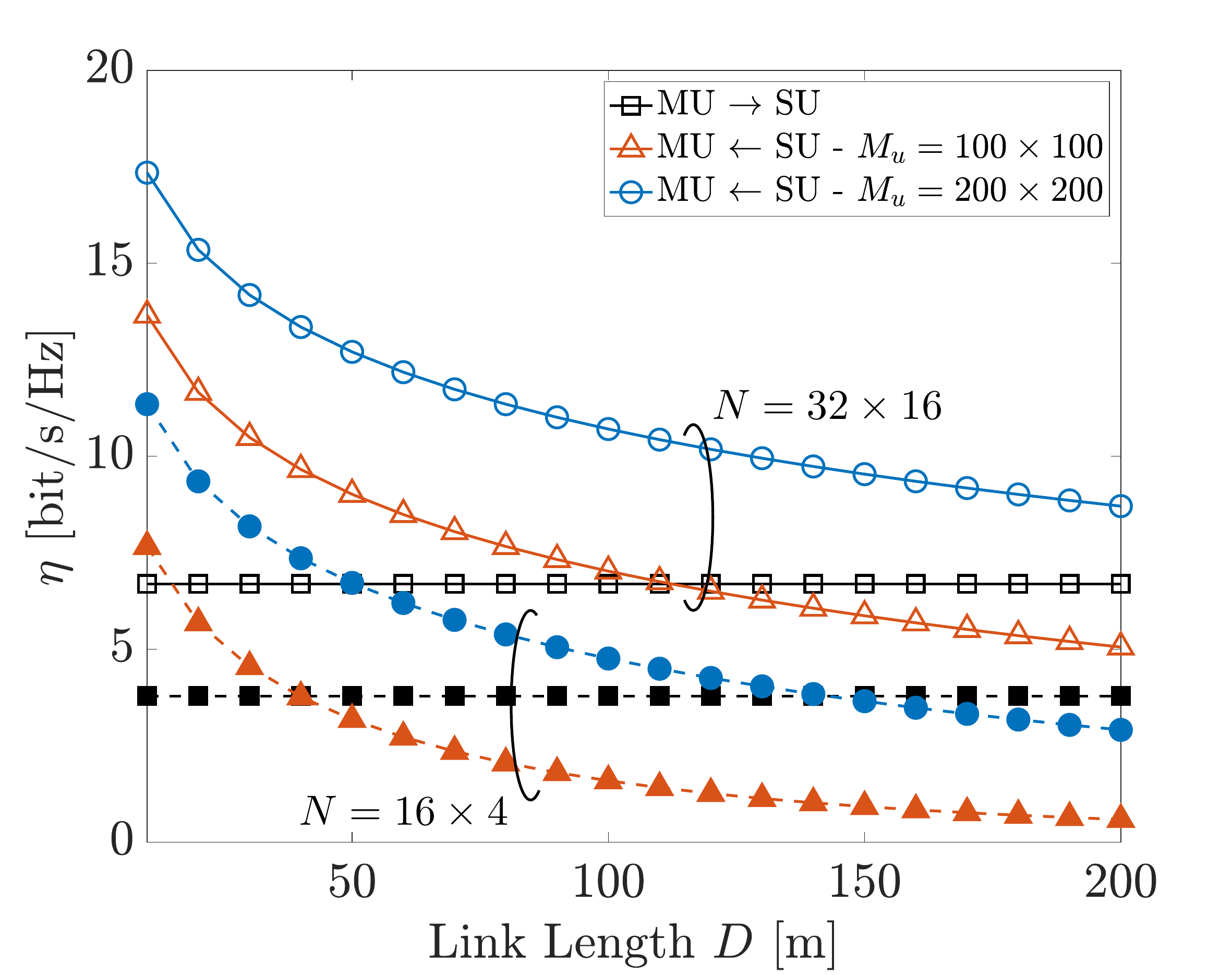} \label{subfig:SE_vs_d}}\\
    \subfloat[]{\includegraphics[width=0.45\textwidth]{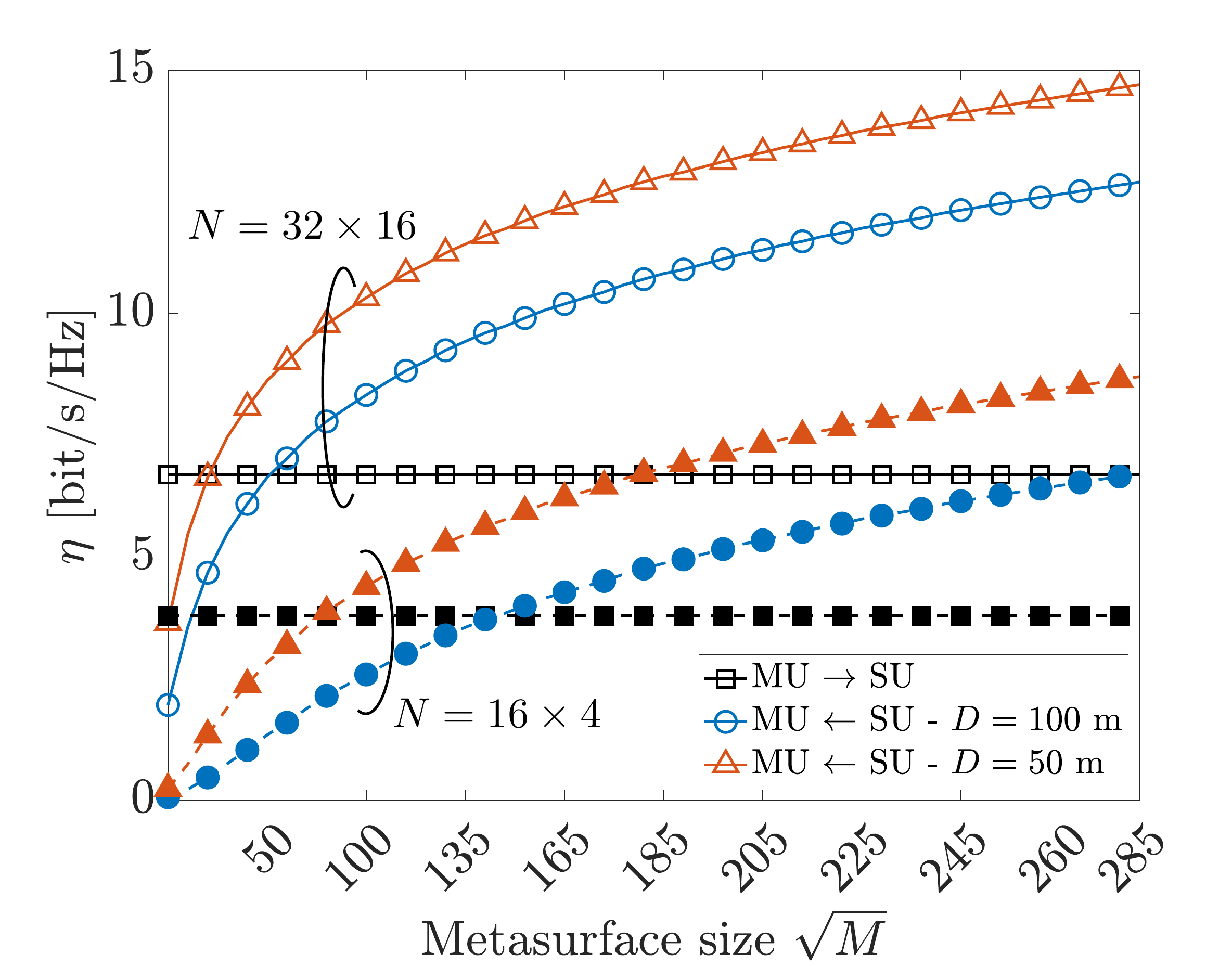} \label{subfig:SE_vs_A}}
    \caption{Spectral efficiency considering different MU arrays configuration ($N = 16 \times 4,\; 32 \times 16$) as a function of various design parameters: (\ref{subfig:SE_vs_SNR}) SNR at the antenna of the SU in MU $\rightarrow$ SU, (\ref{subfig:SE_vs_d}) MU $\leftrightarrow$ SU distance, and (\ref{subfig:SE_vs_A}) STMM size.}
    \label{fig:Design}
\end{figure}

This section presents numerical and analytical results for the proposed STMM-based wireless communication system, aimed at gaining insight into upper achievable performance, operating conditions, and possible limitations. In all the following, the frequency of operation is $f_0=30$ GHz. We consider a CPM modulation with a rectangular PSF, $L=1$, and $M$-ary$=2$, where $\kappa=h/(2T_u)$. This coincides with CP-FSK for $h=1$ (herein considered) and with MSK for $h=1/2$. The trends and considerations are the same for a generic CPM or any phase/frequency modulation, but the absolute values must be derived case by case.

\subsection{Impact of system parameters}\label{subsect:results_vs_param}

The first set of results, summarized in Fig. \ref{fig:Design}, shows the achievable SE (both MU $\rightarrow$ SU and MU $\leftarrow$ SU) varying system parameters, in the utter performance assumption of $|AF(\theta, \phi, \kappa)|^2 = 1$, namely with \eqref{eq:snr_u}, which is commonly assumed in the state of the art on STMM \cite{9133266}. These results can be attained for (\textit{i}) architecture A with space-time decoupling \eqref{eq:ADecoupling}, (\textit{ii}) architecture B with space-time decoupling \eqref{eq:BDecoupling} and limited bandwidth $B_u$, and (\textit{iii}) normal incidence, i.e., $\theta_i = 90$ deg.
In all cases, we have $\mu=0.5$, thus the same bandwidth for MU $\rightarrow$ SU and MU $\leftarrow$ SU. Fig. \ref{subfig:SE_vs_SNR} shows the SE varying the MU $\rightarrow$ SU equivalent single-input single-output (SISO) SNR, namely 
\begin{equation}\label{eq:SISO_SNR}
    \Upsilon_0 =\frac{\Upsilon_d}{N M_d} = \frac{\sigma^2_{s,d}}{\varrho_d \,\sigma_z^2},
\end{equation}
as well as the number of MU antennas $N$ and STMM meta-atoms $M_u$. The MU $\curvedrightleft$ SU link length is fixed to $D=100$ m and the number of SU Rx meta-atoms $M_d = 8 \times 4$. We notice the dramatic impact of the number of STMM meta-atoms $M_u$ that allows boosting the MU $\leftarrow$ SU SE (up to be higher than the MU $\rightarrow$ SU one for $M_u=200\times 200$). A similar effect, though to a lesser extent, is obtained by increasing $N$. The trend of the SE varying the MU $\curvedrightleft$ SU link length $D$ is instead shown in Fig. \ref{subfig:SE_vs_d}. For fixed MU $\rightarrow$ SU performance $\eta_d$, the UL one rapidly decays with $D$, still being higher than the former one for massive MU arrays and STMM (i.e., a comparatively high number of STMM meta-atoms). As can be expected, when $N$ is fixed, e.g., by technological limitations, the overall MU $\leftarrow$ SU SE $\eta_u$ is mostly ruled by the number of STMM meta-atoms $M_u$, as confirmed by Fig. \ref{subfig:SE_vs_A}, where the SE is let vary as a function of the STMM size (side number of meta-atoms $\sqrt{M_u}$), for $D=50$ m and $D=100$ m. The results in Fig. \ref{fig:Design} show that with an appropriate design, the STMM can provide excellent  MU $\leftarrow$ SU performance even at relatively large distances, comparable with the cell radius in current mmWave 5G systems. Remarkably, most of the energy employed to achieve $\eta_d$ and $\eta_u$ is spent at the MU, while the SU does not need a power amplifier. However, wideband effects due to the  MU $\leftarrow$ SU signal $\gamma(t)$ practically limit the performance and, mostly, the MU $\leftarrow$ SU bandwidth $B_u$, as discussed in the following.

\begin{figure}[b!]
    \centering
    \subfloat[$M_u = 100\times100$]{\includegraphics[width=0.47\textwidth]{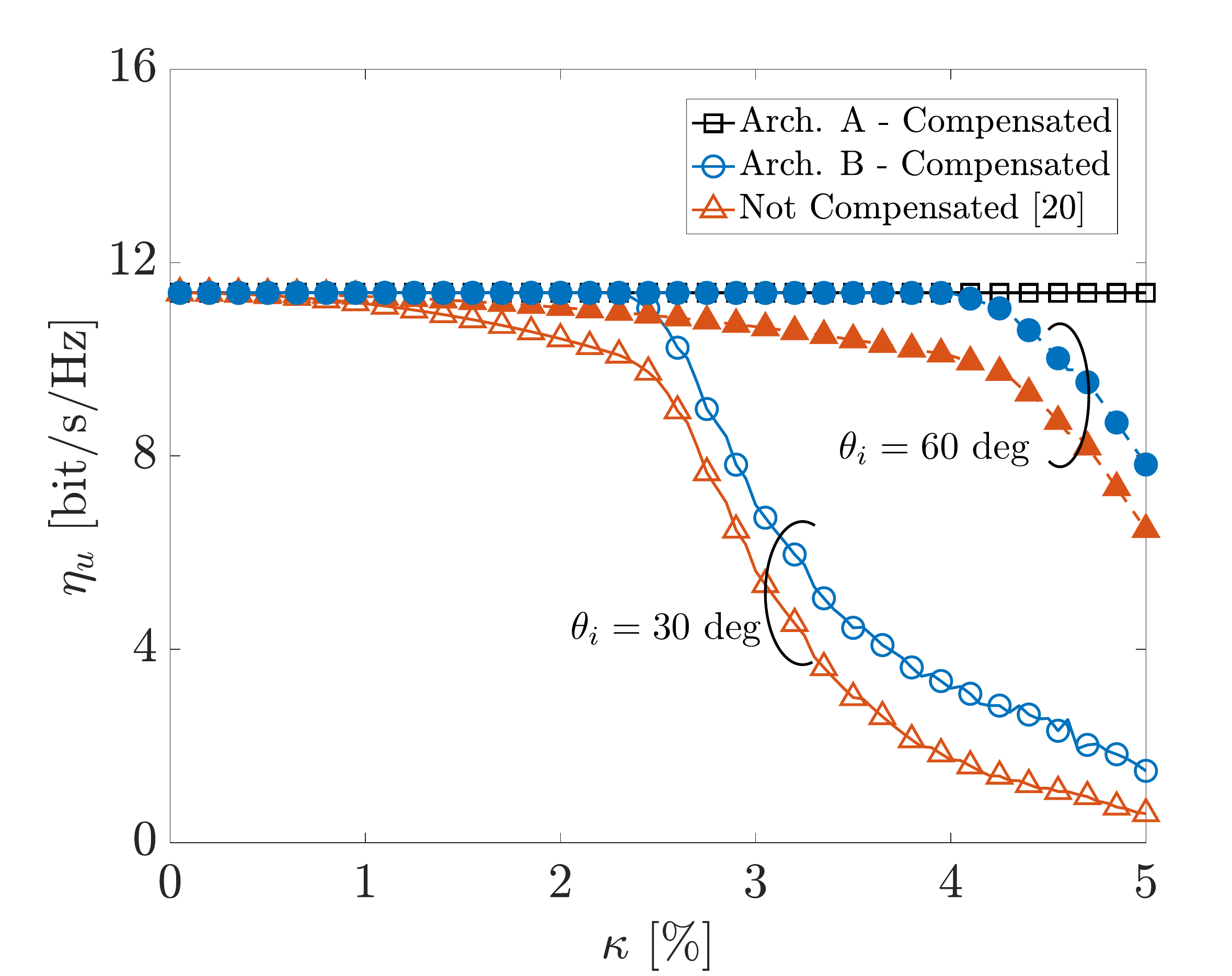}\label{subfig:SEvsBM100}}\\
    \subfloat[$M_u = 200\times200$]{\includegraphics[width=0.47\textwidth]{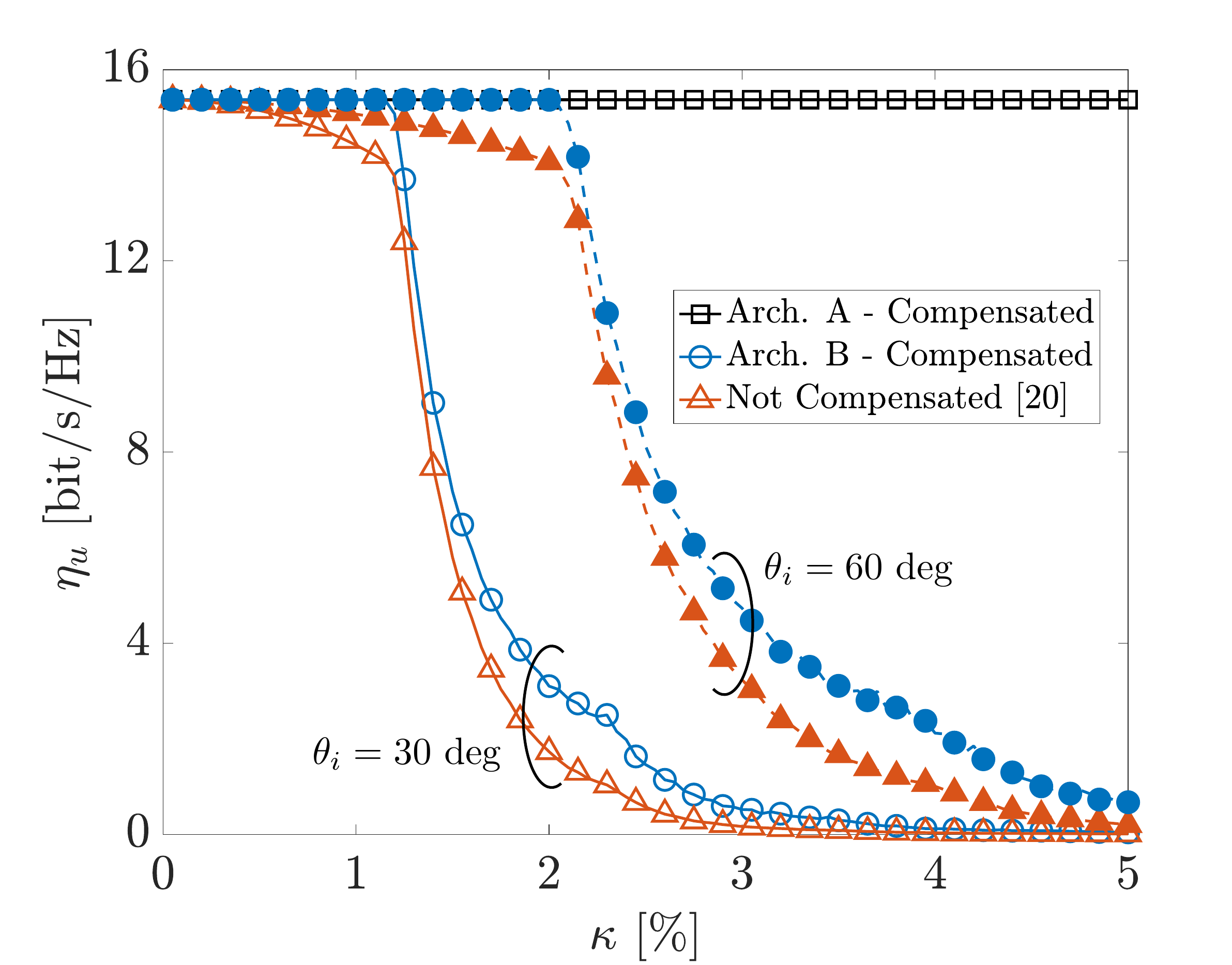}\label{subfig:SEvsBM200}}
    \caption{Spectral efficiency for \ref{subfig:SEvsBM100}) $M_u=100\times 100$ and \ref{subfig:SEvsBM200}) $M_u=200\times200$, and considering different incident angles varying the net frequency shift $\kappa$, which is proportional to the MU $\leftarrow$ SU bandwidth $B_u$. The red line represents the results obtained according to \cite{9133266}, which neglects the compensation. The blue and black lines represent the achievable spectral efficiency $\eta_u$ when compensating the coupling in Arch. B and A, respectively.}
    \label{fig:SE_vs_Kappa}
\end{figure}

\subsection{Impact of space-time phase coupling and decoupling}\label{subsect:results_vs_decoupling}

The wideband effects on the STMM operation are investigated in Fig. \ref{fig:SE_vs_Kappa}, reporting the MU $\leftarrow$ SU SE $\eta_u$ varying $\kappa$, i.e., the equivalent linear frequency shift due to the modulating signal $\gamma(t)$. We assume \textit{(i)} $B_u\gg B_d$, such that to ignore any wideband effect due to the MU $\rightarrow$ SU signal $s_d(t)$, that is the approximation (a) in \eqref{eq:rxSignalExpanded} and \textit{(ii)} perfect channel state information (CSI) is available at the SU, namely the incidence/reflection angles $(\theta,\phi)$ are perfectly known. The distance is $D=100$ m and $N=32\times 16$. The aim of these results is to quantitatively examine the impact of the space-time phase coupling at the STMM, while the impact on an imperfect CSI is discussed in the next set of results. The analysis in Section \ref{subsect:SGradientDistortion} has shown that the space-time coupling is not influenced by transmitted power or distance, but rather by MU $\leftarrow$ SU bandwidth $B_u$, angle of incidence/reflection $(\theta,\phi)$, and the size of the STMM $M_u$. Fig. \ref{fig:SE_vs_Kappa} considers both architectures A and B, with ideal compensations \eqref{eq:ADecoupling}-\eqref{eq:BDecoupling}, respectively, and two incidence angles, $\theta=30,60$ deg ($\phi=0$ deg for simplicity). The non-compensation case is shown as a benchmark. architecture A has the best performance, independent of $\kappa$, thanks to the possibility of compensating for the two-way delay at each meta-atom of the STMM. Considering architecture B, instead, where only one time-varying component is present (and common to all the STMM meta-atoms), the UL SE $\eta_u$ progressively degrades for increasing grazing incidence angles ($\theta\rightarrow 0$ deg) and STMM meta-atoms $M_u$. In both Fig. \ref{subfig:SEvsBM100} and \ref{subfig:SEvsBM200}, we can distinguish three operation regions, outlined in Appendix \ref{app:SNR} and corresponding to \textit{(i)} $\Delta T \ll T_u$, i.e., ideal narrowband STMM operation; \textit{(ii)} $\Delta T \leq T_u$, i.e., wideband STMM operation without ISI; \textit{(iii)} $\Delta T > T_u$, i.e., wideband STMM operation with ISI. The first operation region occurs for very low modulation bandwidths (e.g., $\kappa \leq 1\%$ for $M_u=100\times100$, $\theta=30$ deg and $\kappa \leq 0.5\%$ for $M_u=200\times200$, $\theta=30$ deg), limiting the effective MU $\leftarrow$ SU throughput. The second operation region happens for larger modulation bandwidths, up to a \textit{cut-off frequency} $\overline{\kappa}$ (e.g., $\kappa \leq \overline{\kappa} = 2.5\%$ for $M_u=100\times100$, $\theta=30$ deg and $\kappa \leq \overline{\kappa} = 1.2\%$ for $M_u=200\times200$, $\theta=30$ deg). In this regime, the STMM benefits from the space-time decoupling \eqref{eq:BDecoupling} to improve the SE $\eta_u$ up to $\approx 3$ bit/s/Hz (for $\kappa \rightarrow \overline{\kappa}$ in Fig. \ref{subfig:SEvsBM200}). Above the cut-off limit $\overline{\kappa}$, the STMM realized with architecture B does not work properly: the SNR after the matched filter $\Upsilon_u$ is abated by the presence of the ISI from past modulation symbols in $\gamma(t)$ (see Appendix \ref{app:SNR}). The performance rapidly drops to zero. It is worth remarking that past symbols represent ISI as far as the simple matched filter approach (before trellis decoding) is adopted. In principle, the Rx can be designed to account for the unwanted additional delay due to the STMM reflection, attaining a similar performance to architecture A; however, this would increase the complexity of the implementation and it is not considered here. Alternatively, a viable solution may be to consider an STMM that combines the advantages of architectures A and B, e.g., by allowing the delay tuning at different subsets of meta-atoms.

\begin{figure}[b!]
    \centering
    \subfloat[][]{\includegraphics[width=0.45\textwidth]{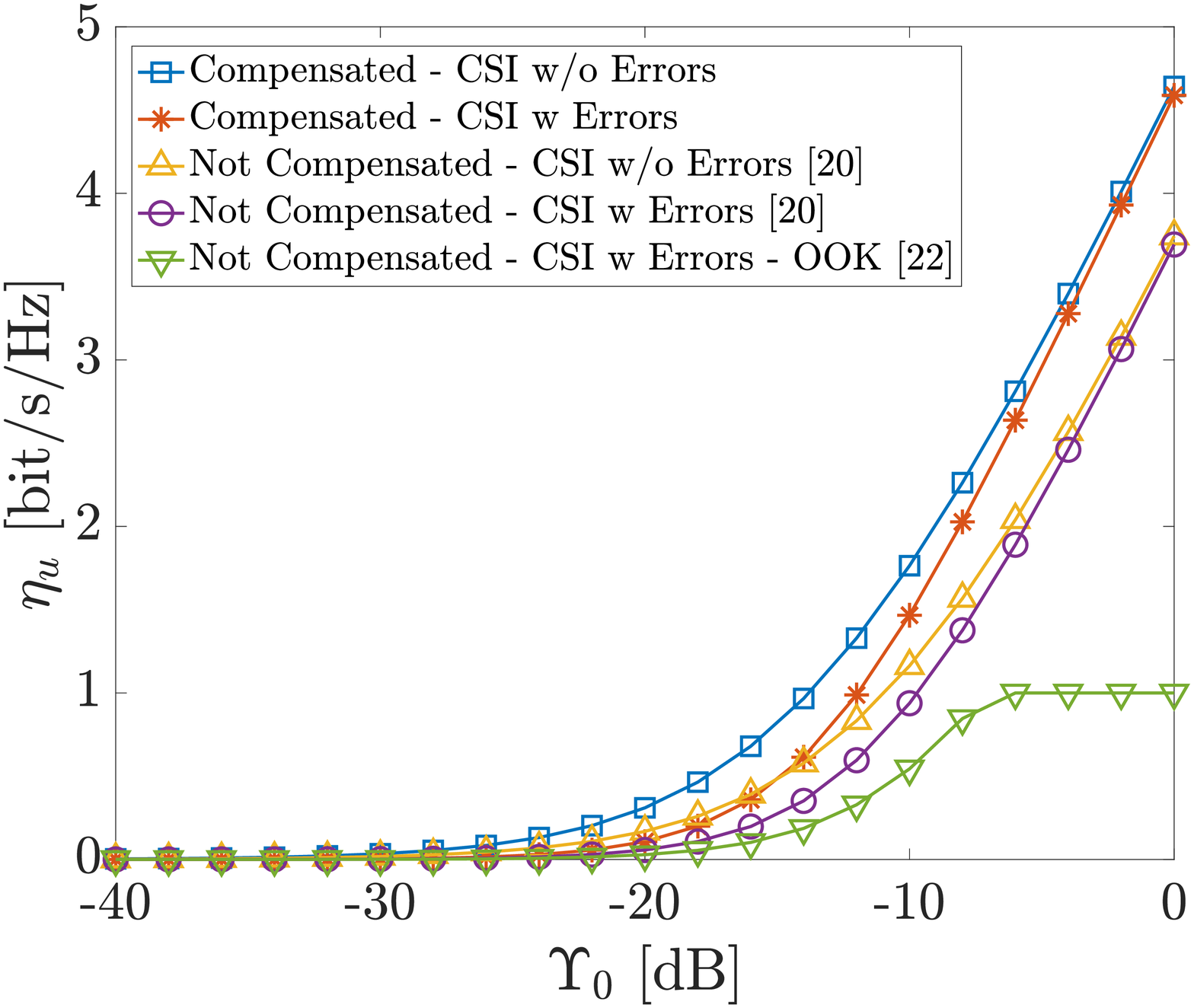}\label{subfig:SE_vs_SNR_chEst_M100}}\\
    \subfloat[][]{\includegraphics[width=0.45\textwidth]{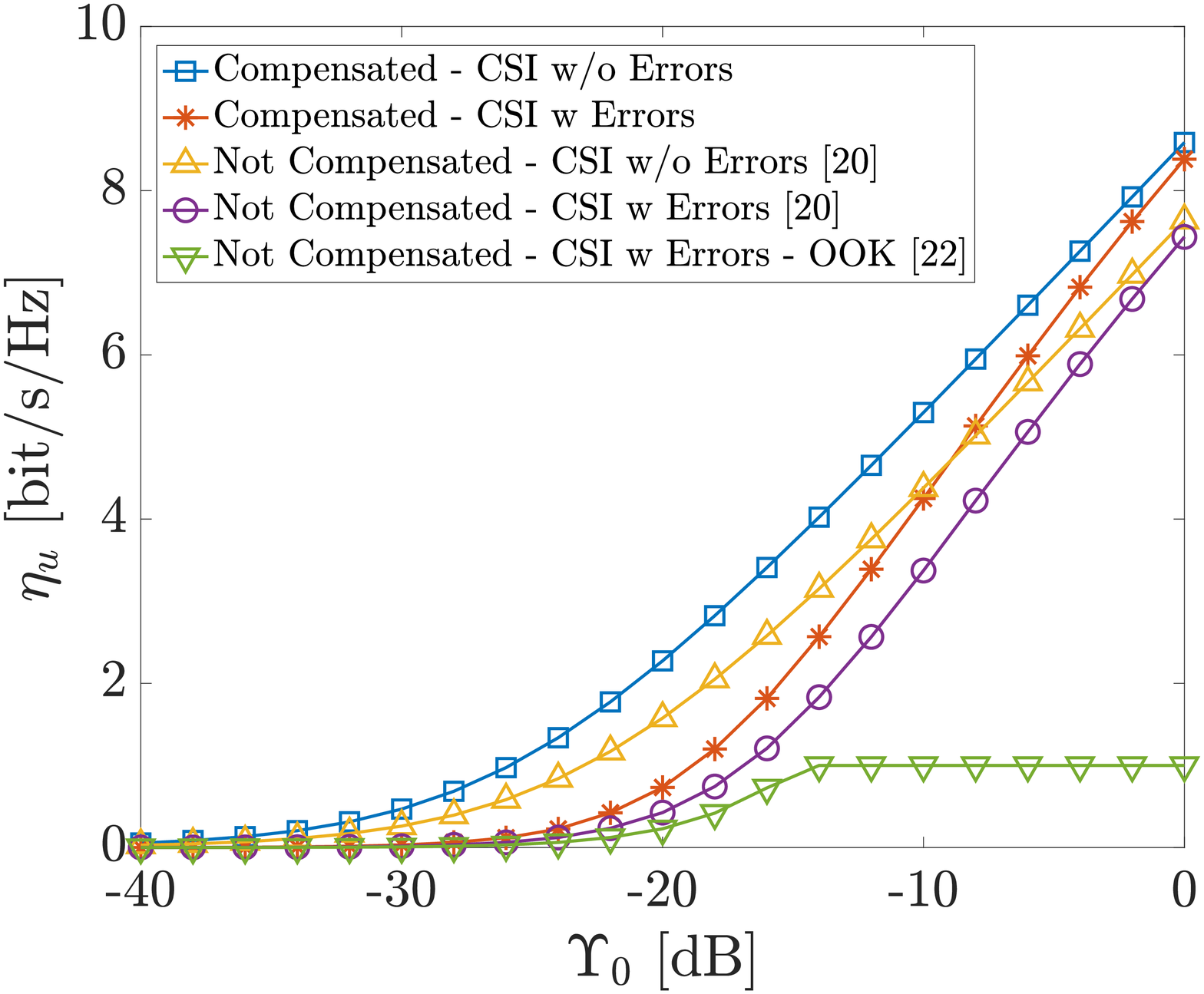}\label{subfig:SE_vs_SNR_chEst_M200}}
    \caption{MU$\leftarrow$SU spectral efficiency varying the equivalent SISO SNR $\Upsilon_0$, for (\ref{subfig:SE_vs_SNR_chEst_M100}) $M_u=100 \times 100$ and (\ref{subfig:SE_vs_SNR_chEst_M200}) $M_u=200 \times 200$ and considering both perfect and imperfect CSI as well as space-time phase coupling (with the proposed compensation and without compensation \cite{9133266}). As a benchmark, green curves denote the spectral efficiency for a OOK-modulated backscatter communication system subject to CSI errors and without compensation (e.g., \cite{Liang2022_RIS_backscatter_survey}).}
    \label{fig:SE_vs_SNR_chEst}
\end{figure}

\subsection{Impact of imperfect CSI}\label{subsect:results_vs_CSI}

In practice, incidence/reflection angles $(\theta,\phi)$ are estimated at the SU Rx array, leading to unavoidable errors affecting both the spatial phase configuration of the STMM (imperfect pointing) as well as the space-time phase decoupling (imperfect compensation). The combination of the latter two effects has the consequence of diminishing the amount of Rx signal energy at the MU side, decreasing the SE w.r.t. the ideal case (see Appendix \ref{app:channel_est} for the analytical derivations). 
Fig. \ref{fig:SE_vs_SNR_chEst} shows the MU$\leftarrow$SU SE $\eta_u$  varying the equivalent SISO SNR $\Upsilon_0$ \eqref{eq:SISO_SNR}, with and without space-time decoupling (as in \cite{9133266}) and in presence of errors on the estimation of the elevation angle $\theta$, whose variance is quantified herein with the Cram\'{e}r-Rao lower bound (CRLB) \cite{Chetty2022_CRB} (Appendix \ref{app:channel_est}). 
We consider the very same simulation parameters used for Fig. \ref{subfig:SE_vs_SNR}, namely $N=16 \times 4$ Tx antennas at the MU, $M_d=8 \times 4$ Rx antennas at the SU, $\theta=30$ deg, $\phi=0$ deg (known), $\kappa = 1\%$ (Fig. \ref{subfig:SE_vs_SNR_chEst_M100}), $\kappa = 2\%$ (Fig. \ref{subfig:SE_vs_SNR_chEst_M200}). With the current settings, the standard deviation of the error on $\theta$ ranges from $8.7$ deg at $\Upsilon_0=-40$ dB to $0.087$ deg at $\Upsilon_0=0$ dB. 
As a further benchmark, we show the constrained SE of a backscatter communication system where the STMM is OOK-modulated, with no space-time phase decoupling and imperfect CSI as well \cite{Liang2022_RIS_backscatter_survey}. 
In the low SNR region, the error on the estimation of $\theta$ (ruled by $\Upsilon_0$) jeopardizes any space-time phase decoupling effort. Namely, the imperfect pointing caused by a wrong setting of the spatial STMM phase (red curves) is much larger than the effect of the space-time phase coupling counterpart (yellow curves). Notice that yellow curves represent the performance of state-of-the-art works on STMM-based wireless systems \cite{9133266}, where $\theta$ is assumed to be perfectly known but the space-time coupling is not compensated. However, as $\Upsilon_0$ grows large, the performance trend swaps: the effect of the imperfect CSI vanishes (red curves attain the blue ones asymptotically) and only an uncompensated space-time phase coupling affects the performance. Of course, as the spatial selectivity of the STMM increases ($M_u$), the impact of imperfect CSI and coupling gets worse. Interestingly, we observe a threshold SNR $\overline{\Upsilon}_0$ value (depending on the size of of the STMM $M_u$) that determines two operating regions: for $\Upsilon_0 \leq \overline{\Upsilon}_0$, the SE $\eta_u$ is dominated by the effect of CSI errors (requiring proper countermeasures, e.g., increasing the Rx array size $M_d$), while for $\Upsilon_0 > \overline{\Upsilon}_0$, the performance is dominated by the space-time phase coupling, confirming the outcomes of Fig. \ref{fig:SE_vs_Kappa} and requiring the suitable compensation proposed in this work. In any case, the proposed STMM-based wireless communication system outperforms the OOK-modulated one, which only employs a two-level amplitude modulation (yielding a -3 dB matched filter gain at the MU) and does not operate any compensation for the space-time phase coupling.

\begin{figure}[b!]
    \centering
    \includegraphics[width=0.45\textwidth]{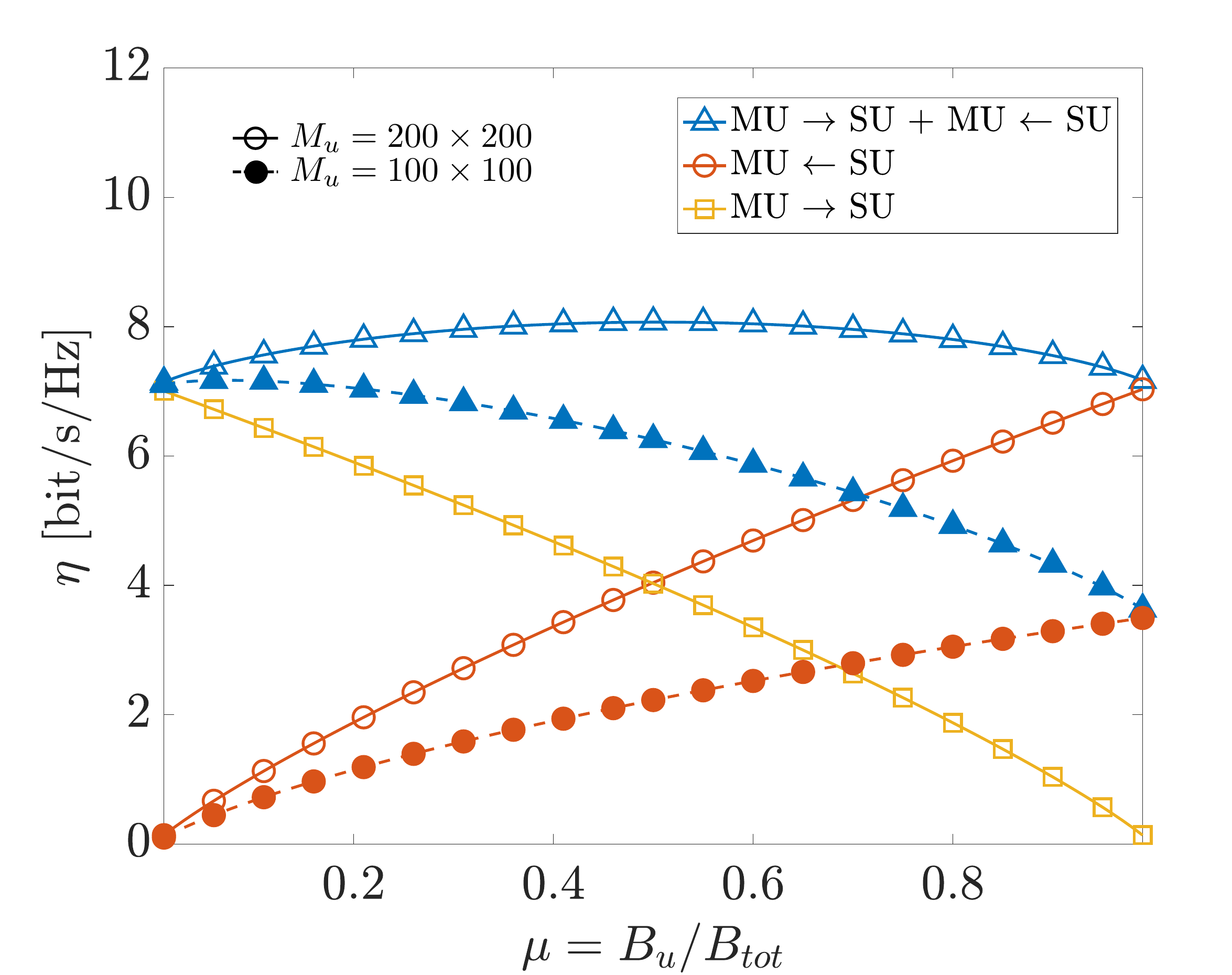}
    \caption{Total spectral efficiency considering different metasurface sizes, i.e., $M_u = 100 \times 100$ and $M_u=200 \times 200$, and different distances varying the bandwidth ratio $\mu$.}
    \label{fig:SE_vs_MU}
\end{figure}

\subsection{Impact of the bandwidth ratio $\mu$}\label{subsect:results_vs_mu}

Concerning the optimal choice of $\mu$, namely the bandwidth splitting between MU $\rightarrow$ SU and MU $\leftarrow$ SU, it depends on the size of the STMM and the distance between the MU and SU. Fig. \ref{fig:SE_vs_MU} shows the MU $\rightarrow$ SU, MU $\leftarrow$ SU, and total SE $\eta_d$, $\eta_u$ and $\eta$ respectively, varying $\mu$ for $D=100$ m, $N=32\times 16$, and assuming ideal narrowband STMM behavior (upper-performance limit). The total SE shows a peak, as predicted by \eqref{eq:D+Lcapacity_derivative}, that depends on the number of STMM meta-atoms $M_u$. This result can have two readings: one can maximize the spectral efficiency $\eta$ over the entire bandwidth $B_{tot}$, and thus the fraction $\mu$ allocated depends on the specific setup, or one can allocate the fraction $\mu$ based on the MU and SU requirements, but the spectral efficiency $\eta$ over the entire band may not be optimal. In other words, the method used to allocate the fraction $\mu$ of the band can be based on either maximizing overall spectral efficiency $\eta$ or on specific MU and SU requirements. While maximizing overall spectral efficiency $\eta$ is desirable in some cases, it may not be the best approach if the MU and SU requirements must be met. In such cases, one may have to compromise on the overall spectral efficiency $\eta$ to meet the requirements of the MU and SU.

\section{Conclusions and Open Challenges} \label{sect:conclusion}

This paper proposes a mathematical model for the design of STMMs in wireless communication systems, where the spatial component of the modulation controls the direction in which the signal is reflected, and the temporal component conveys information in the designated direction. The model examines two potential implementations, Architectures A and B, with Architecture A having greater design flexibility than Architecture B, which is more cost-effective. The model describes the unique features of the STMMs, including the space-time phase coupling and how it affects system performance. Analytical descriptions of these effects are provided, along with suggested countermeasures. Additionally, the proposed model addresses the design criteria and challenges of STMM-based wireless communications, specifically in the context of a full-duplex system architecture, which is one of the promising solutions where STMMs can be beneficial. The full-duplex system architecture comprises a master unit (MU) that transmits a downlink MU $\rightarrow$ SU signal and a slave unit (SU) that can receive, reflect, and modulate the MU $\rightarrow$ SU using STMM technology. In order to transmit information in uplink MU $\leftarrow$ SU, the SU modulates the temporal phase gradient applied to the metasurface meta-atoms. To validate the analytical findings, numerical results are presented, demonstrating the efficacy of the proposed model, even in the presence of imperfect CSI acquisition. 
Nonetheless, the proposed mathematical model only serves as an initial step for designing and analyzing STMM-based wireless communication systems. Further research is imperative to fully explore the capabilities and limitations of STMM technology. This should include investigations into multi-user scenarios, the impact of mobility on system performance, and experimental validations to bridge the gap between theory and practical implementation. Addressing these challenges will pave the way for harnessing the full potential of STMMs in future wireless communication applications.

\appendices
\section{}\label{app:SNR}
 The Rx signal at the MU side undergoes matched filtering and phase estimation, before being fed to the sequence detector, e.g., Viterbi. We derive the SNR after the matched filter at the MU in the assumption that the Rx signal is the one expressed by \eqref{eq:rxSignalExpanded}, approximation $(a)$. The estimated phase at the $n$-th MU $\leftarrow$ SU symbol is provided by 
\begin{equation}\label{eq:MF}
    \hat{\gamma}(nT_u) = \arg \left\{\int_{n T_u}^{(n+1)T_u} y_u(t) s_d^*(t-2\tau) dt\right\}.
\end{equation}
The SNR is obtained from $r_n$ with a double expectation over the Tx signal $s_d(t)$ and the modulating phase $\gamma(t)$. 
The power of the useful signal part is in \eqref{eq:Rx_power_w_MF},
\begin{figure*}[!t]
\begin{equation}\label{eq:Rx_power_w_MF}
\begin{split}
     \sigma^2_{signal} & = \mathbb{E}_{s,\gamma}\left[\bigg\lvert\int_{T_u} \rho \lvert s_d(t-2\tau)\rvert ^2 \sum_{u,v} e^{j\gamma(t-\Delta t_{u,v} -\tau)} \bigg\rvert^2dt \right]
     \leq \\
     &  \leq \mathbb{E}_s\left[\int_{ T_u} \hspace{-0.25cm}|\rho|^2 \lvert s_d(t-2\tau)\rvert ^4 dt\right] \mathbb{E}_\gamma \left[  \int_{ T_u}  \bigg\lvert \sum_{u,v} e^{j\gamma(t-\Delta t_{u,v} -\tau)} \bigg\rvert^2 dt \right] \approx\\
     & \approx |\rho|^2 \sigma^4_{s,d} T_u \,\,\mathbb{E}_\gamma \left[  \int_{T_u}  \bigg\lvert \sum_{u,v} e^{j\gamma(t-\Delta t_{u,v} -\tau)} \bigg\rvert^2 dt \right],
\end{split}
\end{equation}
\hrulefill
\end{figure*}
where we made use of the Schwartz inequality and assume, for simplicity of exposition, that the MU $\rightarrow$ SU signal $s_d(t)$ has a constant power $\sigma^2_{s,d}$ within the integration interval $T_u$. According to the size of the STMM $M_u = M_{u,x}\times M_{u,y}$ compared to the MU $\leftarrow$ SU bandwidth $B_u$ (related to the symbol time through \eqref{eq:cmpBandwidth}), we can distinguish three cases, as detailed in the following.
\subsubsection{$\Delta T \ll T_u$} In this case, the STMM exhibits a narrowband behavior w.r.t. to the MU $\leftarrow$ SU signal, namely the variation of $\gamma(t)$ across the STMM can be neglected. Therefore, we have 
\begin{equation}
    \mathbb{E}_\gamma \left[  \int_{T_u}  \bigg\lvert \sum_{u,v} e^{j\gamma(t-\Delta t_{u,v} -\tau)} \bigg\rvert^2 dt \right] \approx M_u^2 T_u 
\end{equation}
and the reflection gain is maximized under the assumption of the perfect knowledge of the channel (i.e., perfect spatial phase configuration in \eqref{eq:cmpBandwidth}). The power of the noise after the matched filter is
\begin{equation}
    \sigma^2_{noise}  \leq N_0 \sigma^2_{s,d} T_u
\end{equation}
(including the beamforming gain at the MU), thus the SNR is upper-bounded by
\begin{equation}
    \Upsilon_u = \frac{\sigma^2_{signal}}{\sigma^2_{noise}} \leq \frac{|\rho|^2 M_u^2 \sigma_{s,d}^2 }{\sigma^2_z} T_u B_{tot}
\end{equation}
from which follows \eqref{eq:snr_u} by expanding $\rho$.
\subsubsection{$\Delta T \leq T_u$} In this case, the STMM exhibits a wideband behavior w.r.t. to the MU $\leftarrow$ SU signal, namely the variation of $\gamma(t)$ across the STMM cannot be ignored in \eqref{eq:Rx_power_w_MF}, but it is less then the symbol duration. Now, we have 
\begin{equation}
\begin{split}
    \mathbb{E}_\gamma \left[  \int_{T_u}  \bigg\lvert \sum_{u,v} e^{j\gamma(t-\Delta t_{u,v} -\tau)} \bigg\rvert^2 dt \right] \approx \\
    \approx M_u^2 |AF(\theta,\phi,\kappa)|^2 T_u 
\end{split}
\end{equation}
it yields the case explored in Section \ref{sec:STCDesign} with the beam squinting effect. The SNR upper bound is
\begin{equation}
    \Upsilon_u \leq \frac{|\rho|^2 M_u^2 \sigma_{s,d}^2 }{\sigma^2_z} |AF(\theta,\phi,\kappa)|^2 T_u B_{tot}.
\end{equation}
The previous expression tends to be the narrowband one when applying compensation techniques in Section \ref{subsect:Decoupling}.

\subsubsection{$\Delta T > T_u$} In this case, the STMM exhibits a wideband behavior w.r.t. to the MU $\leftarrow$ SU signal and the variation of $\gamma(t)$ across the STMM exceeds the symbol duration. This means that the MU receives a signal with multiple superposed symbols, and the STMM meta-atoms can be clustered according to the specific MU $\leftarrow$ SU symbol they are reflecting at a given instant of time. Considering for simplicity a linear STMM, we have $C = \lceil M_u \Delta t/T_u\rceil$ spanned symbols, where $\Delta t$ is the propagation delay across two adjacent meta-atoms, $\Delta t = (\lambda_i/4)(\cos\theta/c)$. Considering a 2D STMM again, we can define the following set of STMM meta-atoms:
\begin{equation}
\begin{split}
     \mathcal{M}'_u &= \{(u,v)\lvert\Delta t_{u,v} < T_u\}\\
    \mathcal{M}''_u &= \{(u,v)\lvert\Delta t_{u,v} \geq T_u\}
\end{split}
\end{equation}
for $u=0,...,M_{u,x}-1$, $v=0,...,M_{u,y}-1$, where the first one indicates the symbol of interest and the second all the remaining $C-1$ ones. Notice that estimating the phase with the simple matched filter in \eqref{eq:MF} implies that the $C-1$ symbols other than the $n$-th of interest act as \textit{inter-symbol interference} (ISI). Therefore, we have the relation in \eqref{eq:signal+ISI}
\begin{figure*}[!t]
\begin{equation}\label{eq:signal+ISI}
    \mathbb{E}_\gamma \left[ \int_{T_u}  \bigg\lvert \sum_{u,v} e^{j\gamma(t-\Delta t_{u,v} -\tau)} \bigg\rvert^2 dt \right] \approx \underbrace{\lvert \mathcal{M}'_u \rvert^2 |AF(\theta,\phi,\kappa)|^2 T_u}_{\text{useful signal}} + \underbrace{\mathbb{E}_\gamma \left[ \int_{T_u}  \bigg\lvert \sum_{(u,v)\in \mathcal{M}''_u} e^{j\gamma(t-\Delta t_{u,v} -\tau)} \bigg\rvert^2 dt \right]}_{\text{ISI}}
\end{equation}
\hrulefill
\end{figure*}
and thus the useful signal power is
\begin{equation}
    \sigma^2_{signal} \approx |\rho|^2 \sigma^4_{s,d} \lvert \mathcal{M}'_u \rvert^2 |AF(\theta,\phi,\kappa)|^2 T^2_u 
\end{equation}
while the ISI power $\sigma^2_{ISI}$ comes from the integration across different symbols, with detrimental effects on system performance. The SNR turns into SINR as
\begin{equation}
 \begin{split}
     \Upsilon_u = \frac{|\rho|^2 \sigma^4_{s,d} \lvert \mathcal{M}'_u \rvert^2 |AF(\theta,\phi,\kappa)|^2 T^2_u}{ \sigma^2_{ISI} + N_0 \sigma^2_{s,d} T_u }.
 \end{split}
\end{equation}

\section{}\label{app:channel_est}

This Appendix derives the model for the characterization of the system performance in case of imperfect CSI, i.e., when the incident elevation angle $\theta$ is estimated with unavoidable error at the SU side, $\widehat{\theta} = \theta + \delta\theta$. Extension to the case in which $\phi$ is subject to estimation errors is possible but covered here, thus we consider $\phi=0$ (and $\widehat{\phi}=0$) to ease the derivations. We herein consider, for simplicity, that \textit{(i)} the estimated angle $\widehat{\theta}$ is Gaussian-distributed, i.e., $\widehat{\theta} \sim \mathcal{N}(\theta, \sigma^2_\theta)$ and \textit{(ii)} the estimation attains the CRLB, hence the variance of the error $\delta\theta$ is 
\cite{Chetty2022_CRB}:
\begin{equation}\label{eq:CRB_theta}
    \sigma^2_\theta = \frac{96}{\pi^2 \sin^2 \theta \, \Upsilon_0 \, N^2 M_{d}(M_{d}^2-1)}.
\end{equation}
Let us consider the Rx signal \eqref{eq:rxSignalExpanded} $(a)$, where we apply the following space-time phase at the STMM (with decoupling as for \eqref{eq:ADecoupling})
\begin{equation}
    \widehat{\beta}_{q}(t) = 4 \pi f_i \widehat{\Delta t}_{q} + \gamma(t + \widehat{\Delta t}_{q})
\end{equation}
where
\begin{equation}
    \begin{split}
    \widehat{\Delta t}_{q} = \frac{\lambda_i}{4 c } q \cos\widehat{\theta} & \approx \frac{\lambda_i}{4 c} q \left(\cos\theta - \delta\theta \sin\theta \right)
    \end{split}
\end{equation} 
(the approximation is for small $\delta \theta$). We obtain from \eqref{eq:rxSignalExpanded}:
\begin{figure*}[t!]
\centering
\begin{equation}\label{eq:rxSignalExpanded_chest2}
    \begin{split}
        y_u(t) & = M_{u,y}\, \rho  \sum_{q=0}^{M_{u,x}-1} e^{-j4 \pi f_i (\Delta t_{q} - \widehat{\Delta t}_{q})} e^{j\gamma(t-(\Delta t_{q}-\widehat{\Delta t}_{q})-\tau)} s_d(t-2\tau) + z(t)\\
        & \overset{(i)}{=}  M_{u,y}\, \rho  \, s_u(t-\tau) s_d(t-2\tau) \sum_{q=0}^{M_{u,x}-1} \underbrace{e^{-j4 \pi f_i (\Delta t_{q} - \widehat{\Delta t}_{q})}}_{\text{imperfect pointing}} \underbrace{e^{-j2 \pi \kappa f_i (\Delta t_{q} - \widehat{\Delta t}_{q})}}_{\text{imperfect decoupling}} + z(t) \\
        & \overset{(ii)}{\approx} M_{u,y}\, \rho  \, s_u(t-\tau) s_d(t-2\tau) \sum_{q=0}^{M_{u,x}-1} 
        e^{j\frac{\pi}{2} (2+ \kappa) q \delta \theta \sin \theta} + z(t)
    \end{split}
\end{equation}
\hrulefill
\end{figure*}
where equality $(i)$ holds by considering a linearly time-varying phase $\gamma(t) = 2 \pi \kappa f_i t$ (or piece-wise linear) and approximation $(ii)$ is for small $\delta\theta$. Notice that \eqref{eq:rxSignalExpanded_chest2} resembles the ideal multiplicative model \eqref{eq:rxSignalExpanded3}, except for the residual summation, representing the loss due to imperfect estimation of $\theta$. We can distinguish between two effects: the imperfect pointing, due to errors in the spatial phase configuration, and the imperfect decoupling, due to errors in the compensation of the residual delay at the STMM. The ensemble average of the energy loss at the MU side due to imperfect CSI can be evaluated from \eqref{eq:rxSignalExpanded_chest2} $(ii)$ as follows:
\begin{equation}\label{eq:AF_2D_chest}
    \begin{split}
      \frac{1}{M_{u,x}^2} \mathbb{E}_{\widehat{\theta}}  \left[ \bigg\lvert\frac{\sin\left(M_{u,x} \,(\pi/4) (2+\kappa) \delta \theta \sin\theta\right)}{\sin\left((\pi/4) (2+\kappa) \delta \theta \sin\theta\right)} \bigg \rvert^2\right]
    \end{split}
\end{equation}
where the expectation is carried out over the distribution of $\widehat{\theta}$, and the overall effects are shown in Section \ref{subsect:results_vs_CSI}.

\bibliographystyle{IEEEtran}
\bibliography{Bibliography}

\end{document}